\documentclass[a4paper,UKenglish,cleveref]{lipics-v2021}

\usepackage[utf8]{inputenc}
\usepackage[T1]{fontenc}
\usepackage{microtype}
\usepackage{amsmath,amssymb,mathtools}
\usepackage{booktabs}
\usepackage{enumitem}
\usepackage{algorithm}
\usepackage{algpseudocode}
\usepackage{float}
\usepackage{array}
\usepackage{graphicx}
\usepackage{xurl}
\usepackage{makecell}

\let\le\leqslant
\let\ge\geqslant

\newcommand{\Floor}[1]{\left\lfloor #1 \right\rfloor}
\newcommand{\1}{\mathbf{1}}
\DeclareMathOperator{\mult}{\mathcal{W}}

\title{Counting All Lattice Rectangles in the Square Grid in Near-Linear Time}

\author{Dmitry Babichev}{France}{dimitry008@gmail.com}{}{}
\author{Sergey Babichev}{France}{bsl84848@gmail.com}{}{}
\authorrunning{Dmitry Babichev and Sergey Babichev}
\category{}
\relatedversion{}
\supplement{}
\funding{}
\acknowledgements{}

\EventLongTitle{}
\EventShortTitle{}
\EventAcronym{}
\EventYear{}
\EventDate{}
\EventLocation{}
\SeriesVolume{}
\ArticleNo{}

\hideLIPIcs

\setlist{topsep=2pt,itemsep=1pt,parsep=0pt,partopsep=0pt}
\begin{document}
	\maketitle
	
	\begin{abstract}
		We study the exact counting problem for all lattice rectangles contained in the square
		$[0,n)\times[0,n)$, including non-axis-parallel ones.
		Starting from the standard parametrization by a primitive direction $(u,v)$ and two side lengths, we derive several exact algorithms: the classical $O(n^2)$ sweep, decompositions of complexity $O(n^{3/2}\log n)$ and $O(n^{4/3}\log n)$, a ten-moment weighted-floor-sum reduction of complexity $O(n\log^3 n)$, and a divisor-layer algorithm with the complexity $O(n\log^2 n)$. We also give an all-values algorithm that computes $F(1),\ldots,F(N)$ in $O(N^{3/2})$ arithmetic operations. The main idea behind the near-linear one-value algorithms is to reduce the geometric summation to constant-size families of weighted floor sums closed under Euclidean-style affine and reciprocal transformations. Besides the exact algorithmic results, we derive a two-term asymptotic expansion,
		$
		F(n)=\frac{4\log 2-1}{\pi^2}n^4\log n+B\,n^4+o(n^4)$ with the explicit formula for $B$, 
		which provides an independent consistency check for the large-$n$ numerical data produced by the algorithms.
	\end{abstract}
	
	\section{Introduction}

	We study the following problem.
	
	\medskip
	\noindent
	\textbf{Problem.}
	Count all rectangles whose vertices lie in the integer lattice and are contained in the square $[0,n)\times[0,n)$, including non-axis-parallel rectangles.
	
	\medskip
	
	Counting lattice configurations in planar grids is a classical topic at the interface of combinatorics, geometry of numbers, and lattice-point enumeration. In this paper we study an exact counting problem for lattice rectangles.
	
	The resulting counting function is the OEIS sequence A085582~\cite{OEISA085582}. Geometrically, it extends the classical axis-parallel rectangle count to arbitrary lattice orientations in the same $n \times n$ grid of points. Arithmetically, each non-axis-parallel family is governed by a primitive direction $(u,v)$, which brings in coprimality, divisor sums, and floor-sum recurrences. Computationally, the OEIS entry already collects exact values, tables, and scripts, so faster exact algorithms are useful both for extending the data and for testing asymptotic predictions.
	
	The classical split appears in the OEIS decomposition
	$
	A085582(n)=A000537(n-1)+A113751(n),
	$
	where $A000537(n-1)=\binom{n}{2}^2$ is the axis-parallel count and $A113751(n)$ records the non-axis-parallel contribution~\cite{OEISA085582}. Our goal is to show that the full sequence is not merely a table of values, but a structured counting problem admitting several increasingly efficient exact algorithms.
	
	We are not aware of prior work on this exact enumeration problem in the present algorithmic form. Our main contribution is an algorithmic reduction from geometric enumeration to a constant-size family of weighted Euclidean floor-sum kernels, which yields the complexity improvements proved below. For related background on floor-sum reciprocity, elementary number-theoretic estimates, and lattice-point enumeration, see~\cite[Section~3.5]{ConcreteMath} and~\cite{HardyWright,BeckRobins,BeckRobinsDedekind,ZagierDedekind}. The use of primitive lattice directions also connects the problem to the broader literature on primitive lattice points in planar domains and polygonal counting problems~\cite{HuxleyNowak,Patrascu}.
	
	Our main contributions are as follows.
	\begingroup\setlength{\itemsep}{2pt}
	\begin{enumerate}
		\item We derive several exact one-value algorithms beyond the quadratic baseline, with complexities $O(n^{3/2}\log n)$, $O(n^{4/3}\log n)$, $O(n\log^3 n)$, and $O(n\log^2 n)$. The $O(n\log^2 n)$ method uses M\"obius divisor layers and a square-root cover inside each layer, while the weighted floor-sum kernels are evaluated by Euclidean recurrences. See \cref{tab:algorithm-summary} for more details.
		\item We also derive an all-values algorithm computing the whole prefix table $F(1),\ldots,F(N)$ in $O(N^{3/2})$ arithmetic operations and $O(N)$ memory. It replaces repeated summatory evaluations by coefficient arrays indexed by the exact threshold $ax+by$ and then applies a M\"obius divisor convolution.
		\begin{table}[ht]
			\centering
			\small
			\setlength{\tabcolsep}{6pt}
			\renewcommand{\arraystretch}{0.98}
			\begin{tabular}{@{}l l c c l@{}}
				\toprule
				Method & Main idea & Time & Authors & Kernel \\
				\midrule
				Baseline &
				\makecell[tl]{primitive-direction\\enumeration} &
				$O(n^2)$ &
				Radcliffe &
				none \\
				\addlinespace
				Square-root &
				\makecell[tl]{small/large split;\\coprime prefixes} &
				$O(n^{3/2}\log n)$ &
				this paper &
				coprime prefixes \\
				\addlinespace
				Cubic-root &
				\makecell[tl]{dual parametrization;\\six moments} &
				$O(n^{4/3}\log n)$ &
				this paper &
				\makecell[tl]{six-moment\\weighted floor-sum} \\
				\addlinespace
				Ten-moment &
				\makecell[tl]{M\"obius inversion;\\weighted floor sums} &
				$O(n\log^3 n)$ &
				this paper &
				\makecell[tl]{ten-moment\\weighted floor-sum} \\
				\addlinespace
				Divisor-layer &
				\makecell[tl]{M\"obius layers;\\square-root cover} &
				$O(n\log^2 n)$ &
				this paper &
				\makecell[tl]{six-moment\\weighted floor-sum} \\
				\addlinespace
				All-values &
				\makecell[tl]{event arrays;\\M\"obius convolution} &
				$O(N^{3/2})$ &
				this paper &
				\makecell[tl]{arithmetic-progression\\updates} \\
				\bottomrule
			\end{tabular}
			\caption{Summary of the algorithms discussed in the paper.}
			\label{tab:algorithm-summary}
		\end{table}

		\item We derive a two-term asymptotic expansion
		$
		F(n)=\frac{4\log 2-1}{\pi^2}n^4\log n+B\,n^4+o(n^4),
		$
		with an explicit constant $B$. This asymptotic analysis is partly independent of the algorithmic development, but it is included here because it gives a stringent large-$n$ validation of the exact counts and in particular of the high-precision computations used in the fastest implementations.
	\end{enumerate}\endgroup

	\medskip
	
	The intermediate algorithms are included because they expose the structural reductions leading to the near-linear methods: the $O(n^{3/2}\log n)$ split introduces the direction/side-length duality, the $O(n^{4/3}\log n)$ algorithm isolates a finite floor-moment state space, the $O(n\log^3 n)$ algorithm uses a ten-state weighted floor-sum kernel, and the $O(n\log^2 n)$ divisor-layer algorithm combines the same six-state kernel with a M\"obius-layer square-root cover. The all-values algorithm uses the same threshold geometry, but stores coefficients by exact threshold so that the whole table is recovered by prefix sums.
	
	\medskip
	
	Throughout the paper we count arithmetic operations in a standard RAM model. In the implementation, all loop indices and intermediate affine parameters are stored in \texttt{std::int128\_t}, which is exact throughout the benchmarked range reported in \cref{sec:experiments}. Accordingly, the stated running times should be read as arithmetic-operation bounds rather than bit-complexity bounds for arbitrarily large integers. The recursive floor-sum kernels have $O(1)$ state size; the remaining memory usage comes from the preprocessing tables.
	
	\section{The standard parametrization}\label{sec:standard}
	This section fixes the basic parametrization and separates the easy axis-parallel and boundary contributions from the primitive non-axis-parallel part.
	
	Let $(u,v)$ be a primitive lattice direction, so $u\ge v > 0$ and $\gcd(u,v)=1$. The orthogonal direction is $(-v,u)$. A rectangle with side lengths $a,b\ge 1$ in these two directions is generated by $a\cdot (u,v)$ and $b\cdot (-v,u)$, hence its axis-aligned bounding box has side lengths $au+bv$ and $av+bu$. Therefore it fits into the half-open square $[0,n)\times[0,n)$ if and only if
	\begin{equation}\label{eq:basic-constraints}
		au+bv\le n,\qquad av+bu\le n.
	\end{equation}
	For such a quadruple $(u,v,a,b)$, the number of placements equals
	\begin{equation}\label{eq:placements}
		(n-au-bv)(n-av-bu).
	\end{equation}
	
	We split
	$
	F(n)=F_0(n)+F_1(n).
	$
	Here \(F_0(n)\) collects the axis-parallel orientations together with the boundary direction \((1,1)\), while \(F_1(n)\) contains all primitive directions with
	\(0<v<u\). The axis-parallel term is the classical rectangular-grid count, the direction \((1,1)\) is the simplest genuinely tilted family, and all remaining work is concentrated in the primitive non-axis-parallel directions. These cases are separated to avoid double counting and to isolate the
	boundary direction \((1,1)\). The axis-parallel rectangles contribute
	\(\binom{n}{2}^2 = n^2(n-1)^2/4\), and the direction \((1,1)\) contributes
	\(n(n-1)^2(n-2)/12\). Hence
	\[
	F_0(n)=\frac{n^2(n-1)^2}{4}+\frac{n(n-1)^2(n-2)}{12}.
	\]
	
	For \(F_1(n)\), by symmetry in \(a,b\), set \(x:=\max(a,b)\) and
	\(y:=\min(a,b)\), so \(x\ge y\ge 1\). The constraints become
	\(xu+yv\le n\) and \(xv+yu\le n\), and the second inequality is redundant because
	\(x\ge y\) and \(u>v\). Thus the two orderings of \((a,b)\) collapse into one
	term; swapping \((u,v)\) with \((v,u)\) only produces the reflected direction,
	excluded by the convention \(u>v>0\). Thus we get
	\begin{equation}\label{eq:xyuv-standard}
		F_1(n)= \!\!\!\!\!\!\!\!\!\!
		\sum_{\substack{u>v>0,\ x\ge y\ge 1\\ \gcd(u,v)=1,\ xu+yv\le n}}\!\!\!\!\!\!\!\!\!\!
		\mult(x,y)\,(n-xu-yv)(n-xv-yu), 
	\end{equation}
	where $\mult(x,y)=2$ if $x=y$ and $\mult(x,y)=4$ if $x>y$. This identity will be the common starting point for the reparametrizations below.

	\section{The classical quadratic algorithm}\label{sec:quadratic}
	We record the natural $O(n^2)$ baseline obtained by sweeping over primitive directions and summing admissible side-length pairs~\cite{OEISA085582,RadcliffeA085582}.
	
	\subsection{Direction sweep}
	
	Fix a primitive direction $(u,v)$ with $u > v$. For each admissible $y$, the variable $x$ runs through an interval determined by \cref{eq:xyuv-standard}. A straightforward implementation computes the contribution of one primitive direction in time
	$
	O\!\left(n/u\right),
	$
	or in time $O\!\left(n/u\cdot \log n\right)$ if one performs coprimality tests online and evaluates the inner arithmetic progressions without the closed-form simplifications.
	
	The number of primitive pairs with first coordinate $u$ is $\varphi(u)+O(1)$, so using the classical estimate
	$
	\sum_{u\le n}\frac{\varphi(u)}{u}=O(n)
	$ we get:
	\[
	\sum_{u\le n}\varphi(u)\frac{n}{u}=O(n^2),
	\qquad
	\sum_{u\le n}\varphi(u)\frac{n\log n}{u}=O(n^2\log n),
	\]

	\begin{proposition}[classical baseline]
		The rectangle-counting problem admits
		\begingroup\setlength{\itemsep}{2pt}
		\begin{itemize}
			\item a direct implementation of complexity $O(n^2\log n)$ if primitivity and inner summation are handled naively, and
			\item an $O(n^2)$ implementation after standard preprocessing of primitive directions and elementary arithmetic simplifications.
		\end{itemize}\endgroup
	\end{proposition}
	
	Algorithmically, the classical baseline sweeps over primitive directions $(u,v)$ and, for each such direction, over the admissible values of the secondary side length parameter; for fixed $(u,v,y)$, the remaining interval in $x$ is summed explicitly. We state this baseline explicitly because every later improvement modifies exactly one bottleneck of this sweep. Compare the standard divisor-summation and hyperbola-method viewpoints in~\cite{MontgomeryVaughan,Tenenbaum}; a more implementation-oriented version is recorded in \cref{app:pseudocode}. For readers coming from the OEIS entry, this is also the closest formalization of the straightforward computation behind the existing tables and scripts~\cite{RadcliffeA085582}.
	
	\section{\texorpdfstring{A square-root decomposition: $O(n^{3/2}\log n)$}{A square-root decomposition}}\label{sec:sqrt}
	
	Let
	$
	B:=\Floor{\sqrt n}.
	$
	We split the set of primitive directions into two regions
	\[
	\mathcal D_{\mathrm{small}}
	=
	\{(u,v): \gcd(u,v)=1,\ u > v > 0,\ u\le B\},
	\]
	\[
	\mathcal D_{\mathrm{large}}
	=
	\{(u,v): \gcd(u,v)=1,\ u > v > 0,\ u>B\}.
	\]
	
	\subsection{Small directions}
	
	For $(u,v)\in\mathcal D_{\mathrm{small}}$, we keep the classical parametrization $(u,v,x,y)$. The contribution of one direction is computed in time
	$
	O\!\left(n/u \cdot \log n\right),
	$
	and hence the total cost is
	\[
	\sum_{u\le B}\varphi(u)\frac{n\log n}{u}
	=
	O(nB\log n)
	=
	O(n^{3/2}\log n).
	\]
	
	\subsection{Large directions}
	If $(u,v)\in \mathcal D_{\mathrm{large}}$, then $u>B$, hence every admissible rectangle has
	$x\le n/u<n/B$. Therefore we switch the order of summation and use the
	parameterization from \cref{sec:standard}. Restricting the formula
	for $F_1(n)$ to $\mathcal D_{\mathrm{large}}$, we obtain
	\[
	F_{1,\mathrm{large}}(n)=
	\sum_{\substack{(u,v)\in \mathcal D_{\mathrm{large}}\\ x\ge y\ge 1,\ xu+yv\le n}}
	\mult(x,y)\,(n-xu-yv)(n-xv-yu).
	\]
	
	Fix $u,x,y$ and set
	\[
	v_{\max}:=\min\!\left(u-1,\left\lfloor \frac{n-xu}{y}\right\rfloor\right).
	\]
	Then the inner summation runs over all $v\in[1,v_{\max}]$ with $\gcd(u,v)=1$, and
	the placement factor is a quadratic polynomial in $v$:
	\[
	(n-xu-yv)(n-xv-yu)=A_0(u;x,y)+A_1(u;x,y)\,v+A_2(x,y)\,v^2,
	\]
	where
	\[
	A_0(u;x,y)=(n-xu)(n-yu),\qquad
	A_1(u;x,y)=-(x(n-xu)+y(n-yu)),\qquad
	A_2(x,y)=xy.
	\]
	Hence for fixed $u,x,y$ we need only the three coprime prefix sums
	\[
	C_j(u;X):=\sum_{\substack{1\le v\le X\\ \gcd(u,v)=1}} v^j,
	\qquad j\in\{0,1,2\}.
	\]
	Indeed, by M\"obius inversion, $\1_{\gcd(u,v)=1}=\sum_{d\mid u,\ d\mid v}\mu(d)$, and therefore
	\[
	C_j(u;X)=\sum_{d\mid u}\mu(d)\sum_{m\le X/d}(dm)^j.
	\]
	Consequently the $v$-sum becomes
	\[
	\sum_{\substack{1\le v\le v_{\max}\\ \gcd(u,v)=1}}(n-xu-yv)(n-xv-yu)
	=A_0C_0(u;v_{\max})+A_1C_1(u;v_{\max})+A_2C_2(u;v_{\max}).
	\]
	So the whole large-direction contribution is
	\[
	F_{1,\mathrm{large}}(n)=
	\sum_{u>B}\ \sum_{x\le n/u}\ \sum_{y\le x}
	\mult(x,y)
	\bigl(A_0C_0(u;v_{\max})+A_1C_1(u;v_{\max})+A_2C_2(u;v_{\max})\bigr).
	\]
	For each fixed $u$, the number of pairs $(x,y)$ is $O((n/u)^2)$, and one coprime-prefix query is answered in $O(\tau(u))$ time from the squarefree divisors of $u$. Summing over $u>B$ gives
	\[
	\sum_{u>B} O\!\left(\tau(u)\Bigl(\frac{n}{u}\Bigr)^2\right)
	=O\!\left(n^2\sum_{u>B}\frac{\tau(u)}{u^2}\right)
	=O\!\left(\frac{n^2\log n}{B}\right).
	\]
	
	\begin{theorem}[square-root decomposition]
		Choosing $B=\Floor{\sqrt n}$ and treating small directions and large directions by dual parametrizations yields an algorithm of complexity
		$
		O(n^{3/2}\log n).
		$
	\end{theorem}
	
	At the algorithmic level, one chooses the threshold $B=\Floor{\sqrt n}$, evaluates the small-direction range $u\le B$ by the one-direction routine from the previous section, and evaluates the complementary range by the dual sweep in the $(x,y)$ variables. \cref{app:pseudocode} records a fuller pseudocode version of this decomposition.
	
	\section{\texorpdfstring{A cubic-root decomposition: $O(n^{4/3}\log n)$}{A cubic-root decomposition}}\label{sec:cuberoot}
	
	The next refinement follows the same philosophy as the square-root split, but with a stronger one-direction kernel. Let
	\[
	B:=\Floor{n^{2/3}}.
	\]
	We again split the set of primitive directions at the threshold $u=B$.
	
	\subsection{\texorpdfstring{Small directions: one primitive pair in $O(\log n)$}{Small directions: one primitive pair}}\label{sec:lognpair}

	Fix a primitive direction $(u,v)$ with
	\[
	u>v\ge 1,\qquad \gcd(u,v)=1.
	\]
	Denote by $c_{u,v}(n)$ its total contribution in the parametrization of~\eqref{eq:xyuv-standard}:
	\[
	c_{u,v}(n)
	:=
	\sum_{\substack{x\ge y\ge 1\\ xu+yv\le n}}
	\mult(x,y)\,(n-xu-yv)(n-xv-yu),
	\]
	We sum over $x$, with $y$ as the inner variable. For fixed $x$, the admissible values of $y$ satisfy
	\[
	1\le y\le \min\!\left(x, M(x)\right),\text{ where }M(x) = \Floor{\frac{n-xu}{v}}
	\]
	Thus the $x$-range splits at the transition point $x=\frac{n}{u+v}$. Therefore we obtain two segments:
	\[
	I_1=\left[1,\Floor{\frac{n}{u+v}}\right],
	\qquad
	I_2=\left[\Floor{\frac{n}{u+v}}+1,\Floor{\frac{n}{u}}\right].
	\]
	
	On the first segment $I_1$, the bound coming from $xu+yv\le n$ is inactive, so $1\le y\le x$.
	On the second segment $I_2$, the active upper bound is
	$
	1\le y\le \Floor{\frac{n-xu}{v}}.
	$
	Thus $c_{u,v}(n)$ is the sum of two segment contributions.
	
	On $I_1$, we separate the diagonal $y=x$ from the strict range $1\le y<x$.
	The diagonal contribution is computable in $O(1)$ time.
	
	On $I_2$, the upper bound for $y$ is $M(x)$, and since $M(x)<x$ throughout this segment, we always have $x>y$, hence $\mult(x,y)=4$. Expanding
	\[
	(n-xu-yv)(n-xv-yu)=P_0(x)+P_1(x)y+P_2y^2,
	\]
	where $P_0(x)$ is quadratic in $x$, $P_1(x)$ is linear in $x$, and $P_2$ is constant, and then carrying out the inner summation over $y$, we obtain a polynomial in $x$ and $M(x)$. More precisely, the total contribution of $I_2$ is a fixed linear combination of the six floor moments
	\begin{equation}\label{eq:sixmoments}
		\sum_{x\in I_2} M(x), \ \ 
		\sum_{x\in I_2} xM(x), \ \ 
		\sum_{x\in I_2} x^2M(x), \ \ 
		\sum_{x\in I_2} M(x)^2, \ \ 
		\sum_{x\in I_2} xM(x)^2, \ \ 
		\sum_{x\in I_2} M(x)^3.
	\end{equation}
	To formalize the required queries, for integers $N\ge 0$, $m\ge 1$, and $a,b\ge 0$, define the six-moment kernel
	\begin{equation}\label{eq:six-kernel-family}
		\mathcal H_{p,q}(N;m,a,b)
		:=
		\sum_{t=0}^{N-1} t^p\Floor{\frac{at+b}{m}}^{q},
		\ \ 
		(p,q)\in\{(0,1),(1,1),(2,1),(0,2),(1,2),(0,3)\}.
	\end{equation}
	The interval sums appearing in \eqref{eq:sixmoments} have negative slope in the summation variable, so before invoking the kernel we first reverse the index.
	
	\begin{lemma}[sign reversal for floor moments]\label{lem:six-sign-reversal}
		Let $m\ge 1$, $u\ge 0$, $B\in\mathbb Z$, and $L\le R$. Put $N:=R-L+1$. Then for every $p\ge 0$ and $q\ge 1$,
		\[
		\sum_{x=L}^{R} x^p \Floor{\frac{B-ux}{m}}^{q}
		=
		\sum_{j=0}^{p} (-1)^j \binom{p}{j} R^{p-j}
		\mathcal H_{j,q}(N;m,u,B-uR).
		\]
	\end{lemma}
	
	\begin{proof}
		Let $x=R-t$, where $0\le t\le R-L$. Then \[\Floor{\frac{B-ux}{m}}=\Floor{\frac{B-uR+ut}{m}},
		\qquad
		x^p=(R-t)^p=\sum_{j=0}^{p} (-1)^j\binom{p}{j}R^{p-j}t^j.
		\]
		Substituting and exchanging the finite sums gives the claim.
	\end{proof}
	
	In particular, each of the six sums in \eqref{eq:sixmoments} is a linear combination of six-moment kernel values with nonnegative slope $a=u$; these six direct states are already recursively closed, since the affine and reciprocal Euclidean steps do not increase the total weighted degree $p+q$, so the family with $q\ge 1$ and $p+q\le 3$ is stable. The exact transition formulas are recorded in \cref{app:six-moments}. The reciprocal step used below is analogous to the classical conjugation/reciprocity viewpoint in integer-partition decompositions~\cite{AndrewsEriksson}.
	
	\begin{lemma}[affine closure of the six-moment kernel]\label{lem:six-affine}
		Each affine normalization step expresses a state $\mathcal H_{p,q}(N;m,a,b)$ as a linear combination of states $\mathcal H_{p',q'}(N;m,a',b')$ inside the same six-state family, together with ordinary polynomial sums.
	\end{lemma}
	
	\begin{proof}
		Write the floor argument as $q_1t+q_0+g(t)$ after removing the easy Euclidean quotients. Expanding $(q_1t+q_0+g(t))^q$ shows that every nontrivial term still has the form $t^{p'}g(t)^{q'}$ with $(p',q')\in\{(0,1),(1,1),(2,1),(0,2),(1,2),(0,3)\}$, while the terms with $q'=0$ are ordinary polynomial sums. The explicit identities are listed in \cref{app:six-moments}.
	\end{proof}
	
	\begin{lemma}[reciprocal closure of the six-moment kernel]\label{lem:six-reciprocal}
		Under the reciprocal Euclidean step, each state of the six-moment family is expressible as a linear combination of states of the same family and polynomial sums.
	\end{lemma}
	
	\begin{proof}
		Write
		\[
		\mathcal H_{p,q}(N;m,a,b)=\sum_{x=0}^{N-1} x^p y(x)^q,
		\qquad y(x):=\Floor{\frac{ax+b}{m}},
		\]
		and assume $0\le a,b<m$. Put
		\[
		Y:=\Floor{\frac{a(N-1)+b}{m}}.
		\]
		Instead of summing by columns indexed by $x$, we regroup the same staircase region by horizontal levels $t=0,1,\dots,Y-1$. The condition $y(x)\ge t+1$ is equivalent to
		\[
		x\ge \Floor{\frac{mt+(m-b-1)}{a}}+1.
		\]
		Hence the transposed staircase is encoded by the reciprocal floor function
		\[
		g(t):=\Floor{\frac{mt+(m-b-1)}{a}},
		\qquad 0\le t<Y,
		\]
		which is exactly the same construction with the roles of $a$ and $m$ interchanged. When one rewrites the sums over the horizontal strips, the lower bounds contribute only polynomial weights in $t$, so every term becomes a linear combination of moments
		\[
		\sum_{t=0}^{Y-1} t^{p'} g(t)^{q'}
		=
		\mathcal H_{p',q'}(Y;a,m,m-b-1),
		\]
		with $(p',q')\in\{(0,1),(1,1),(2,1),(0,2),(1,2),(0,3)\}$, together with ordinary power sums. Thus the same six-state family is preserved under the reciprocal step, and the larger Euclidean parameter decreases from $m$ to $a<m$. The explicit coefficient identities are listed in \cref{app:six-moments}.
	\end{proof}
	
	\begin{corollary}[evaluation of the six-moment kernel]\label{cor:floor-kernel}
		For fixed integers $(n,m,a,b)$, the six moments in \eqref{eq:sixmoments} can be computed in time
		$
		O(\log n)
		$
		by an Euclidean recursion. The recursive system is closed: no moments outside the family \eqref{eq:sixmoments} are generated.
	\end{corollary}
	
	\begin{proof}
		By \cref{lem:six-affine,lem:six-reciprocal}, the recursion alternates between two Euclidean-style operations: affine normalization replaces $(a,b)$ by their residues modulo $m$, and the reciprocal step then swaps the active pair $(m,a)$ to $(a,m)$ with $a<m$. Thus after each nontrivial cycle the larger Euclidean parameter strictly decreases, exactly as in the ordinary Euclidean algorithm. The recursion therefore has depth $O(\log m)=O(\log n)$, and each step performs only $O(1)$ arithmetic operations on the constant-size family of moments. Hence the six-moment kernel is evaluable in $O(\log n)$ time.
	\end{proof}
	
	On the first segment, all required quantities reduce to polynomial sums, hence are computable in $O(1)$ time. On the second segment, \cref{lem:six-sign-reversal,cor:floor-kernel} gives an $O(\log n)$ evaluation.
	
	\begin{corollary}\label{cor:one-direction}
		For a fixed primitive direction $(u,v)$ with $u\ge v\ge 1$, the contribution $c_{u,v}(n)$ can be computed in time
		$
		O(\log n).
		$
	\end{corollary}
	
	Since the number of primitive pairs with $u\le B$ is $O(B^2)$, the total cost of the small part is
	$
	O(B^2\log n)=O(n^{4/3}\log n).
	$
	
	\subsection{Large directions}
	
	We now consider the complementary region $u>B=n^{2/3}$. Geometrically, this is completely analogous to the large-direction analysis in \cref{sec:sqrt}. The difference is only that the cutoff is now $B=n^{2/3}$ and that, instead of the square-root sweep, we evaluate the resulting one-dimensional sums by coprime prefix sums.
	\[
	\sum_{u>B} O\!\left(\tau(u)\Bigl(\frac{n}{u}\Bigr)^2\right)
	=O\!\left(n^2\sum_{u>B}\frac{\tau(u)}{u^2}\right)
	=O\!\left(\frac{n^2\log n}{B}\right) = O(n^{4/3}\log n).
	\]

	Thus the large-direction part matches the cost of the small-direction part, but it uses a different kernel: coprime prefix sums of orders $0,1,2$, not floor moments.
	
	\begin{theorem}[cubic-root decomposition]
		The lattice-rectangle counting problem admits an algorithm of complexity
		$
		O(n^{4/3}\log n).
		$
	\end{theorem}
	
	Fix $B=\Floor{n^{2/3}}$. For $u\le B$, each primitive direction is reduced to $O(1)$ queries to the six-moment kernel; for $u>B$, the dual parametrization is handled by the coprime prefix sums $C_0,C_1,C_2$. The main text proves the reduction and the resulting complexity bound, while \cref{app:pseudocode} and \cref{app:six-moments} record the corresponding implementation details and explicit kernel transitions.
	
	\section{\texorpdfstring{A ten-moment reduction to weighted floor sums: $O(n\log^3 n)$}{A ten-moment reduction to weighted floor sums}}\label{sec:final}
	
	We now turn to the ten-moment weighted-floor-sum reduction. This section explains how the quantities $R_{u,y,d}$ arise and how they reduce to weighted floor sums. We first isolate the required weighted floor-sum queries and derive the global complexity conditionally on an $O(\log n)$ query bound; this temporary assumption is discharged below by proving that the relevant kernel is recursively closed and evaluable in $O(\log n)$ time.
	
	\subsection{From primitive directions to M\"obius inversion}
	
	We start from the contribution of the non-axis-parallel directions, restricting to
	$u>v\ge 1$. The coprimality condition is removed by the Möbius identity.
	Writing $v=dt$ with $d\mid u$, $t\ge 1$, and $dt<u$, we obtain
	\[
	F_1(n)
	=
	\sum_{u=2}^{n}
	\sum_{\substack{d\mid u\\ \mu(d)\neq 0}}
	\mu(d)
	\sum_{1\le t<u/d}
	\sum_{\substack{x\ge y\ge 1\\ xu+y\,dt\le n}}
	\mult(x,y)\,(n-xu-y\,dt)(n-x\,dt-yu).
	\]
	Now fix $u$, $y$, and $d$. The remaining inner contribution depends only on these
	three outer parameters, so we define
	\begin{equation}\label{eq:Ruyd}
		R_{u,y,d}
		:=
		\sum_{1\le t<u/d}
		\sum_{\substack{x\ge y\ge 1\\ xu+y\,dt\le n}}
		\mult(x,y)\,(n-xu-y\,dt)(n-x\,dt-yu).
	\end{equation}
	Since necessarily $1\le y\le \lfloor n/u\rfloor$, the whole sum becomes
	\begin{equation}\label{eq:F1-outer}
		F_1(n)
		=
		\sum_{u=2}^{n}
		\sum_{1\le y\le \lfloor n/u\rfloor}
		\sum_{d\mid u}\mu(d)\,R_{u,y,d}.
	\end{equation}
	
	\subsection{How many outer triples?}
	
	Let $T(n)$ be the number of admissible outer triples $(u,y,d)$, where
	$1\le y\le \lfloor n/u\rfloor$ and $d\mid u$ is squarefree. For fixed $u$, the
	number of squarefree divisors of $u$ equals $2^{\omega(u)}$, where $\omega(u)$ is
	the number of distinct prime divisors of $u$. Thus
	\[
	T(n)=\sum_{u\le n}2^{\omega(u)}\Bigl\lfloor\frac{n}{u}\Bigr\rfloor.
	\]
	
	For the algorithmic analysis we only need the upper bound $T(n)=O(n\log^2 n)$. Indeed, this upper bound follows from
	$\lfloor n/u\rfloor\le n/u$, which gives
	$T(n)\le n\sum_{u\le n}2^{\omega(u)}/u$. Next, using
	$2^{\omega(m)}=\sum_{d\mid m}\mu^2(d)$, where $\mu^2(d)$ is the indicator of
	squarefree integers, we obtain
	\[
	\sum_{u\le n}\frac{2^{\omega(u)}}{u}
	=
	\sum_{d\le n}\frac{\mu^2(d)}{d}\sum_{m\le n/d}\frac{1}{m}.
	\]
	Now $\sum_{m\le X}1/m=O(\log X)$ and $\sum_{d\le n}\mu^2(d)/d=O(\log n)$, so
	$\sum_{u\le n}2^{\omega(u)}/u=O(\log^2 n)$ and therefore
	$
	T(n)=O(n\log^2 n).
	$
	Consequently, if each set $R_{u,y,d}$ can be processed in time $O(\log n)$ then the total running time is $O(n\log^3 n)$. The remainder of this section proves exactly this $O(\log n)$ evaluation bound by showing that all required weighted floor sums belong to a finite recursively closed Euclidean kernel.

	\subsection{Weighted floor-sum reduction}\label{sec:wfloors}
	
	The purpose of this subsection is to isolate the one-dimensional arithmetic kernel behind the ten-moment algorithm. The explicit affine and reciprocal transition identities for this kernel are somewhat lengthy and are therefore deferred to \cref{app:ten-moments}; here we state the kernel, explain why it is the right one, and derive the global $O(n\log^3 n)$ bound from the structural properties proved there.
	
	Fix $u,y,d$. By definition,
	\[
	R_{u,y,d}
	=
	\sum_{t<u/d}
	\sum_{\substack{x\ge y\ge 1\\ xu+y\,dt\le n}}
	\mult(x,y)\,(n-xu-y\,dt)(n-x\,dt-yu).
	\]
	For each fixed $x\ge y$, the variable $t$ runs over
	\[
	1\le t\le T(x),
	\qquad
	T(x):=\min\!\left(\Bigl\lfloor \frac{u-1}{d}\Bigr\rfloor,\,
	\Bigl\lfloor \frac{n-ux}{yd}\Bigr\rfloor\right).
	\]
	Write
	\[
	T_0:=\Bigl\lfloor \frac{u-1}{d}\Bigr\rfloor,
	\qquad
	G(x):=\Bigl\lfloor \frac{n-ux}{yd}\Bigr\rfloor.
	\]
	Then the summation over $x$ naturally splits into two zones:
	\begingroup\setlength{\itemsep}{2pt}
	\begin{itemize}
		\item the \emph{capped zone}, where $G(x)\ge T_0$ and therefore $T(x)=T_0$;
		\item the \emph{active-floor zone}, where $1 \le G(x)<T_0$ and therefore $T(x)=G(x)$.
	\end{itemize}\endgroup
	The split point is determined by the inequality $n-ux\ge ydT_0$, that is,
	\[
	x\le X_0:=\Bigl\lfloor \frac{n-ydT_0}{u}\Bigr\rfloor.
	\]
	Hence the capped zone contributes only polynomial sums in $x$, while the active-floor zone is exactly where the weighted floor-sum kernel appears. More precisely,
	\[
	R_{u,y,d}
	=
	\sum_{y\le x\le X_0}\mult(x,y)\,S_x^{\mathrm{cap}}
	+
	\sum_{x> X_0}\mult(x,y)\,S_x^{\mathrm{act}}, \quad \text{where}
	\]
	\[
	S_x^{\mathrm{cap}}:=\sum_{1\le t \le T_0}(n-xu-y\,dt)(n-xdt-yu),
	\qquad
	S_x^{\mathrm{act}}:=\sum_{1\le t \le G(x)}(n-xu-y\,dt)(n-xdt-yu).
	\]
	In the capped zone the inner sum is a polynomial in $x$ because the upper limit $T_0$ is constant, and in the active-floor zone the same expansion reduces everything to moments of $G(x)=\Floor{(n-ux)/(yd)}$. Expanding the summand as a quadratic polynomial in $t$, we obtain
	\[
	(n-xu-y\,dt)(n-xdt-yu)=P_0(x)+P_1(x)t+P_2(x)t^2, \quad \text{where}
	\]
	\[
	P_0(x)=(n-xu)(n-yu),\qquad
	P_1(x)=-d\bigl(y(n-yu)+x(n-xu)\bigr),\qquad
	P_2(x)=xyd^2.
	\]
	Using the identities
	\[
	\sum_{t=1}^{T}1=T,\qquad
	\sum_{t=1}^{T}t=\frac{T(T+1)}{2},\qquad
	\sum_{t=1}^{T}t^2=\frac{T(T+1)(2T+1)}{6},
	\]
	we see that each inner sum is a polynomial in $x$ and in its upper limit $T_0$ or $G(x)$. Thus the direct expansion expresses
	$R_{u,y,d}$ through the following seven basic weighted floor sums:
	\[
	\begin{aligned}
		&\sum T(x),\quad \sum xT(x),\quad \sum x^2T(x),
		&\sum T(x)^2,\quad \sum xT(x)^2,\quad \sum x^2T(x)^2,\quad \sum xT(x)^3.
	\end{aligned}
	\]
	Now $\mult(x,y)=4-2\,\mathbf 1_{x=y}$, so the diagonal contribution $x=y$ contributes only a constant-size correction.
	
	\begin{lemma}[direct weighted floor-sum reduction]\label{lem:weighted-floor-reduction}
		For fixed $(u,y,d)$, the quantity $R_{u,y,d}$ is a linear combination of the seven basic sums displayed above over the active-floor zone $x>X_0$, together with ordinary polynomial sums in $x$ coming from the capped zone $y\le x\le X_0$ and the diagonal correction $x=y$.
	\end{lemma}
	
	The seven direct states are not closed under the Euclidean transitions. We therefore enlarge them to the ten-moment family from \cref{app:ten-moments}. For integers $N\ge 0$, $m\ge 1$, and $a,b\ge 0$, define
	\begin{equation}\label{eq:main-floor-family}
		\mathcal H_{p,q}(N;m,a,b)
		:=
		\sum_{x=0}^{N-1}
		x^p
		\Floor{\frac{ax+b}{m}}^{q},
		\qquad q \ge 1,
		\quad p \ge 0,
		\quad p + q \leqslant 4.
	\end{equation}
	Equivalently, the state space is indexed by
	\[
	(p,q)\in\{(0,1),(1,1),(2,1),(3,1),(0,2),(1,2),(2,2),(0,3),(1,3),(0,4)\}.
	\]
	As in the six-state kernel discussed earlier, the affine and reciprocal Euclidean steps do not increase the total weighted degree $p+q$; since the seven direct queries all satisfy $p+q\le 4$, their closure is therefore contained in the ten-state family above, while the five cases with $q=0$ are ordinary monomial sums treated separately. Since the floor term has negative slope in $x$, we first reverse the index. Thus for any interval $[L,R]$,
	\[
	\sum_{x=L}^{R} x^p\Floor{\frac{N-ux}{yd}}^{q}
	=
	\sum_{j=0}^{p}(-1)^j\binom{p}{j}R^{p-j}
	\mathcal H_{j,q}(R-L+1;yd,u,N-uR).
	\]
	
	\begin{lemma}[affine closure of the weighted kernel]\label{lem:weighted-affine}
		For the evaluation of \eqref{eq:main-floor-family}, each affine normalization step expresses a state $\mathcal H_{p,q}(N;m,a,b)$ as a linear combination of states $\mathcal H_{p',q'}(N;m,a',b')$ with $q'\ge 1$, $p'\ge 0$, $p'+q'\le 4$, together with polynomial sums.
	\end{lemma}
	
	\begin{lemma}[reciprocal closure of the weighted kernel]\label{lem:weighted-reciprocal}
		Under the reciprocal Euclidean step, each state of the extended family is expressible as a linear combination of states of the same family and polynomial sums.
	\end{lemma}
	
	\begin{corollary}[evaluation of the weighted kernel]\label{cor:weighted-kernel-log}
		The ten-moment kernel \eqref{eq:main-floor-family} is evaluable in $O(\log n)$ time.
	\end{corollary}
	
	The proofs follow the same Euclidean-recursion scheme as in \cref{lem:six-affine,lem:six-reciprocal,cor:floor-kernel}. The same affine and reciprocal Euclidean steps preserve the enlarged ten-state family; only the explicit coefficient formulas are longer, so they are deferred to \cref{app:ten-moments}.
	
	\begin{corollary}[evaluation of outer contributions]\label{cor:oracle}
		Each outer contribution $R_{u,y,d}$ is computable in time $O(\log n)$.
	\end{corollary}
	
	\begin{proof}
		By Lemma~\ref{lem:weighted-floor-reduction}, each $R_{u,y,d}$ reduces to a constant number of shifted instances of the ten-moment kernel \eqref{eq:main-floor-family} plus polynomial sums. By Corollary~\ref{cor:weighted-kernel-log}, each such kernel query costs $O(\log n)$ time, while the polynomial-zone contribution is evaluated in $O(1)$ time.
	\end{proof}
	
	Combining the $O(\log n)$ evaluation of each outer contribution $R_{u,y,d}$ with the $O(n\log^2 n)$ bound on the number of admissible outer triples yields the stated complexity bound.
	For an explicit implementation, we also precompute the M\"obius values and the squarefree divisor lists
	\[
	D(u)=\{(d,\mu(d)):d\mid u,\ d\text{ squarefree}\},\qquad u\le n.
	\]
	These lists can be generated from a smallest-prime-factor sieve; their total size is
	$\sum_{u\le n}2^{\omega(u)}=O(n\log n)$. Thus the preprocessing time and storage are $O(n\log n)$ and are dominated by the main $O(n\log^3 n)$ summation.
	
	\begin{theorem}[ten-moment complexity]
		After $O(n\log n)$ preprocessing and using $O(n\log n)$ storage for the squarefree divisor lists, the lattice-rectangle counting problem admits an exact $O(n\log^3 n)$ algorithm.
	\end{theorem}
	
	\begin{proof}
		There are $O(n\log^2 n)$ admissible triples $(u,y,d)$, and each corresponding contribution $R_{u,y,d}$ is computable in $O(\log n)$ time by Corollary~\ref{cor:oracle}. Summing over all outer triples proves the stated running time. The squarefree-list preprocessing has size and construction time $O(n\log n)$, so it does not change the asymptotic time bound; it accounts for the stated storage.
	\end{proof}
	
	Algorithmically, the ten-moment method iterates over the admissible outer triples $(u,y,d)$, evaluates each contribution $R_{u,y,d}$ by the recursive weighted floor-sum kernel, and accumulates it with weight $\mu(d)$. The next two sections give the faster divisor-layer one-value algorithm and the all-values algorithm; see \cref{app:pseudocode,app:ten-moments} for the corresponding pseudocode and full transition identities for this ten-moment stage.

	\section[A divisor-layer algorithm for one value: O(n log squared n)]{A divisor-layer algorithm for one value: \texorpdfstring{$O(n\log^2 n)$}{O(n log^2 n)}}\label{sec:one-value}
	The ten-moment algorithm of \cref{sec:final} inserts M\"obius inversion inside a larger summation over directions and side lengths. The final one-value improvement comes from reversing this order. We keep the M\"obius divisor fixed, remove the coprimality condition for the whole layer at once, and then apply the same square-root principle used in the earlier decompositions. Thus the proof has three parts: a divisor-layer identity, a disjoint square-root cover inside one primitive-free layer, and a reduction of the required layer moments to the six-state floor-sum kernel.
	
	\subsection{Divisor layers}
	
	Let $F_1(n)$ denote the non-axis-parallel contribution, so that $F(n)=F_0(n)+F_1(n)$ and $F_0(n)$ is the closed-form term handled earlier. In the standard variables,
	\begin{equation}\label{eq:F1-one-value-start}
		F_1(n)=
		\sum_{\substack{u>v>0,\ x\ge y\ge 1\\ \gcd(u,v)=1,\ xu+yv\le n}}
		\mult(x,y)\,(n-xu-yv)(n-xv-yu).
	\end{equation}
	Here, as before, $\mult(x,y)=2$ for $x=y$ and $\mult(x,y)=4$ for $x>y$. We now insert
	\[
	\mathbf 1_{\gcd(u,v)=1}=\sum_{d\mid u,\ d\mid v}\mu(d).
	\]
	All sums are finite, so the order of summation may be changed. In the summand with fixed $d$, write $u=da$ and $v=db$. Then $xu+yv\le n$ becomes $ax+by\le \lfloor n/d\rfloor$. Hence
	\begin{equation}\label{eq:F1-divisor-layers}
		F_1(n)=\sum_{d\le n}\mu(d)S_d(n),
	\end{equation}
	where $N_d=\lfloor n/d\rfloor$ and
	\begin{equation}\label{eq:Sd-def}
		S_d(n)=
		\sum_{\substack{a>b\ge1,\ x\ge y\ge1\\ ax+by\le N_d}}
		\mult(x,y)
		\bigl(n-d(ax+by)\bigr)\bigl(n-d(bx+ay)\bigr).
	\end{equation}
	This is the same contribution as in \eqref{eq:F1-one-value-start}: summing $\mu(d)$ over common divisors of $u$ and $v$ restores exactly the primitive-direction condition. The advantage is that the inner layer \eqref{eq:Sd-def} is primitive-free. Its weight is the quadratic polynomial
	\begin{equation}\label{eq:weight-expanded-one-value}
		\bigl(n-d(ax+by)\bigr)\bigl(n-d(bx+ay)\bigr)
		=n^2-nd(a+b)(x+y)
		+d^2\bigl(ab(x^2+y^2)+(a^2+b^2)xy\bigr).
	\end{equation}
	Thus each layer reduces to quadratic moments over the region $a>b$, $x\ge y$, $ax+by\le N_d$.
	
	\subsection{A square-root cover inside one layer}
	
	Fix $d$, write $N=N_d$, and put $B=\lfloor\sqrt N\rfloor$. Since $ax\le ax+by\le N$, an admissible tuple cannot have both $a>B$ and $x>B$. We use the disjoint form of this square-root cover,
	\begin{equation}\label{eq:one-value-disjoint-cover}
		\{ax+by\le N,\ a>b,\ x\ge y\}
		=
		\{x\le B\}\ \sqcup\ \{a\le B,\ x>B\}.
	\end{equation}
	It is equivalent to the symmetric cover $\{a\le B\}\cup\{x\le B\}$, but avoids an explicit inclusion--exclusion step. The first part fixes a small side pair and sums over all direction pairs; the second fixes a small direction pair and keeps only the large-side range $x>B$, which is imposed by subtracting capped side moments. For $1\le q<p\le B$ and $0\le i+j\le2$, define
	\begin{equation}\label{eq:Mij-def-one-value}
		M_{ij}(p,q;N)=
		\sum_{\substack{X\ge Y\ge1\\ pX+qY\le N}}X^iY^j,
		\qquad
		T=\left\lfloor\frac{N}{p+q}\right\rfloor,
		\qquad
		D_j(T)=\sum_{t=1}^T t^j.
	\end{equation}
	Here $p,q$ denote the fixed pair, and $X,Y$ denote the remaining pair being summed.
	Let
	\begin{equation}\label{eq:Phi-def-one-value}
		\begin{aligned}
			\Phi_{p,q}^{(d)}
			={}&n^2M_{00}-nd(p+q)(M_{10}+M_{01})\\
			&+d^2pq(M_{20}+M_{02})+d^2(p^2+q^2)M_{11},
		\end{aligned}
	\end{equation}
	where the moments are evaluated at $(p,q;N)$. Thus $\Phi_{p,q}^{(d)}$ is the unweighted sum of \eqref{eq:weight-expanded-one-value} after fixing one ordered non-diagonal pair equal to $(p,q)$. Let
	\begin{equation}\label{eq:Delta-def-one-value}
		\Delta_{p,q}^{(d)}
		=n^2D_0(T)-2nd(p+q)D_1(T)+d^2(p+q)^2D_2(T).
	\end{equation}
	We also use capped moments
	\[
	M_{ij}^{\cap}(p,q;N,B)=
	\sum_{\substack{B\ge X\ge Y\ge1\\ pX+qY\le N}}X^iY^j
	\]
	and define $\Phi_{p,q}^{(d),\cap}$ from \eqref{eq:Phi-def-one-value} with $M_{ij}$ replaced by $M_{ij}^{\cap}$. The capped diagonal term $\Delta_{p,q}^{(d),\cap}$ is obtained from \eqref{eq:Delta-def-one-value} with $T$ replaced by $\min(B,\lfloor N/(p+q)\rfloor)$.
	
	There remains the diagonal-side part of the first set in \eqref{eq:one-value-disjoint-cover}. For $x=y=t\le B$ put
	\begin{equation}\label{eq:side-diagonal-one-value}
		C_t^{(d)}
		=2\sum_{s\le N/t}
		\left\lfloor\frac{s-1}{2}\right\rfloor(n-dts)^2.
	\end{equation}
	Here $s=a+b$, and $\lfloor(s-1)/2\rfloor$ is the number of pairs $a>b\ge1$ with sum $s$.
	
	\begin{lemma}[one layer formula]\label{lem:assembled-layer}
		For every $d\le n$,
		\begin{equation}\label{eq:layer-assembled-one-value}
			S_d(n)=
			\sum_{1\le q<p\le B}
			\left(8\Phi_{p,q}^{(d)}-6\Delta_{p,q}^{(d)}
			-4\Phi_{p,q}^{(d),\cap}+2\Delta_{p,q}^{(d),\cap}\right)
			+\sum_{t=1}^{B}C_t^{(d)}.
		\end{equation}
	\end{lemma}
	\begin{proof}
		Use the disjoint cover \eqref{eq:one-value-disjoint-cover}. First consider the part $x\le B$. If the fixed side is non-diagonal, say $(x,y)=(p,q)$ with $p>q$, then the direction variables satisfy $a>b$ and therefore the diagonal $a=b$ must be removed completely. The quadratic weight \eqref{eq:weight-expanded-one-value} is symmetric under simultaneously exchanging the fixed pair and the summed pair, so the required moment is the same $\Phi_{p,q}^{(d)}$. Thus the non-diagonal side contribution is $4\Phi_{p,q}^{(d)}-4\Delta_{p,q}^{(d)}$. The diagonal sides in this same part have $x=y=t\le B$ and contribute the terms $C_t^{(d)}$.
		
		It remains to count the second disjoint part, where $a\le B$ and $x>B$. Fix the direction $(a,b)=(p,q)$. Without the restriction $x>B$, summing over all sides gives $4\Phi_{p,q}^{(d)}-2\Delta_{p,q}^{(d)}$, because the side multiplier is $4$ for $X>Y$ and $2$ for $X=Y$. The excluded subrange $x\le B$ is exactly the capped contribution $4\Phi_{p,q}^{(d),\cap}-2\Delta_{p,q}^{(d),\cap}$. Hence this second part contributes
		\[
		4\Phi_{p,q}^{(d)}-2\Delta_{p,q}^{(d)}
		-4\Phi_{p,q}^{(d),\cap}+2\Delta_{p,q}^{(d),\cap}.
		\]
		Adding it to the non-diagonal and diagonal pieces of the $x\le B$ part gives \eqref{eq:layer-assembled-one-value}.
	\end{proof}
	
	\subsection{Moment evaluation and complexity}
	
	It remains to show that the moments in \eqref{eq:layer-assembled-one-value} are cheap. For example,
	\begin{equation}\label{eq:Mij-floor-reduction-one-value}
		M_{ij}(p,q;N)
		=\sum_{Y=1}^{\lfloor N/(p+q)\rfloor}Y^j
		\left(P_i\!\left(\left\lfloor\frac{N-qY}{p}\right\rfloor\right)-P_i(Y-1)\right),
	\end{equation}
	where $P_i(R)=\sum_{1\le X\le R}X^i$. Since $i\le2$, expanding $P_i$ expresses every non-polynomial term as a constant-size linear combination of sums
	\begin{equation}\label{eq:six-state-needed-one-value}
		\sum_{0\le z<m}z^\alpha
		\left\lfloor\frac{Az+C}{M}\right\rfloor^\beta,
		\qquad
		(\alpha,\beta)\in\{(0,1),(1,1),(2,1),(0,2),(1,2),(0,3)\},
	\end{equation}
	up to harmless shifts of the summation interval. This is precisely the six-state kernel of \cref{app:six-moments}. Capped moments are treated the same way: in
	\[
	M_{ij}^{\cap}(p,q;N,B)=
	\sum_{Y=1}^{\min(B,\lfloor N/(p+q)\rfloor)}Y^j
	\left(P_i\!\left(\min\!\left(B,\left\lfloor\frac{N-qY}{p}\right\rfloor\right)\right)-P_i(Y-1)\right),
	\]
	the minimum is removed by splitting the range of $Y$ at $\lfloor(N-pB)/q\rfloor$, with empty ranges ignored. One side of the split is polynomial, and the other is another instance of \eqref{eq:six-state-needed-one-value}. Finally, each $C_t^{(d)}$ is a sum over $s$ of a quadratic polynomial times $\lfloor(s-1)/2\rfloor$; splitting by the parity of $s$ evaluates it in $O(1)$ time.
	
	\begin{lemma}[one layer cost]\label{lem:one-layer-cost}
		For fixed $d$, the layer $S_d(n)$ can be evaluated in $O(N_d\log N_d)$ arithmetic operations and $O(N_d)$ working memory.
	\end{lemma}
	\begin{proof}
		There are $O(B^2)=O(N_d)$ pairs $1\le q<p\le B$. For each pair, the uncapped and capped moment reductions above use only a constant number of six-state kernel calls, and each such call costs $O(\log N_d)$ by \cref{cor:floor-kernel,app:six-moments}. The diagonal-side terms contribute $O(B)$ closed-form evaluations. Therefore the layer cost is $O(N_d\log N_d)$.
	\end{proof}
	
	\begin{theorem}[one value]\label{thm:one-value}
		A single value $F(n)$ can be computed exactly in $O(n\log^2 n)$ arithmetic operations and $O(n)$ memory.
	\end{theorem}
	\begin{proof}
		By \eqref{eq:F1-divisor-layers}, it is enough to compute $S_d(n)$ for $d\le n$ and combine the results with the weights $\mu(d)$. Using \cref{lem:one-layer-cost}, the total time is
		\[
		\sum_{d\le n}O\!\left(\frac nd\log\frac nd\right)=O(n\log^2 n).
		\]
		The values of $\mu(d)$ are obtained by a linear or Eratosthenes-type sieve. At any moment we store only this array and the temporary data for one layer, so the memory consumption is $O(n)$. The closed form for $F_0(n)$ is then added in constant time.
	\end{proof}
	
	\begin{remark}[quotient grouping]\label{rem:quotient-grouping}
		The implementation may group all divisors with the same quotient $N_d=\lfloor n/d\rfloor$. For fixed $N_d$, the layer formula is a quadratic polynomial in $d$, so prefix sums of $\mu(d)$, $d\mu(d)$, and $d^2\mu(d)$ combine such blocks at once. This is only a constant-factor improvement and is not used in the asymptotic bound.
	\end{remark}
	
	\section[All values in O(N to the 3/2)]{All values in \texorpdfstring{$O(N^{3/2})$}{O(N^(3/2))}}\label{sec:all-values}
	
	We now compute the whole table $F(1),F(2),\ldots,F(N)$ in one run. This section follows the same divisor-expanded viewpoint as \cref{sec:one-value}, but replaces repeated one-value summation by event arrays. A primitive-free quadruple is stored at the exact threshold where it first becomes active; after all thresholds have been filled, prefix sums recover every value of the sequence. As before, only the non-axis-parallel part $F_1$ is generated explicitly, and the closed form for $F_0(n)$ is added at the end.
	
	\subsection{Event arrays and recovery}
	
	For a primitive-free quadruple $a>b\ge1$, $x\ge y\ge1$, put
	\[
	L=ax+by,
	\qquad
	K=bx+ay.
	\]
	After applying the M\"obius divisor $d$, this quadruple contributes exactly for those $n$ with $n\ge dL$, and its contribution at such an $n$ is
	\[
	\mu(d)\mult(x,y)(n-dL)(n-dK).
	\]
	Thus the activation threshold is $dL$; the value of $K$ only enters the polynomial coefficient. For every event with threshold at most $N$, add
	\begin{align*}
		E_0[dL]&\mathrel{+}=\mu(d)\mult(x,y),\\
		E_1[dL]&\mathrel{+}=\mu(d)d\mult(x,y)(L+K),\\
		E_2[dL]&\mathrel{+}=\mu(d)d^2\mult(x,y)LK.
	\end{align*}
	Let
	\[
	P_i(n)=\sum_{t\le n}E_i[t].
	\]
	Then
	\begin{equation}\label{eq:all-values-recover}
		F_1(n)=n^2P_0(n)-nP_1(n)+P_2(n),
	\end{equation}
	and $F(n)=F_0(n)+F_1(n)$.
	
	\begin{lemma}[event recovery]\label{lem:event-recovery}
		If the arrays $E_0,E_1,E_2$ contain exactly the event contributions above for all thresholds at most $N$, then \eqref{eq:all-values-recover} gives the correct value of $F_1(n)$ for every $1\le n\le N$.
	\end{lemma}
	\begin{proof}
		For a fixed quadruple and divisor, the term contributes to $F_1(n)$ if and only if $dL\le n$. Summing all events with threshold at most $n$ gives
		\[
		\sum_{dL\le n}\mu(d)\mult(x,y)
		\left(n^2-nd(L+K)+d^2LK\right),
		\]
		which is exactly $n^2P_0(n)-nP_1(n)+P_2(n)$ by the definitions of the three event arrays.
	\end{proof}
	
	\subsection{Separating the M\"obius convolution}
	
	First build divisor-free coefficient arrays $G_0,G_1,G_2$ indexed by $L$. This separates the geometric construction of the coefficients from the arithmetic M\"obius convolution, which is applied only once:
	\begin{align*}
		G_0[L]&=\sum_{ax+by=L}\mult(x,y),\\
		G_1[L]&=\sum_{ax+by=L}\mult(x,y)(L+K),\\
		G_2[L]&=\sum_{ax+by=L}\mult(x,y)LK,
	\end{align*}
	where the sums are over $a>b\ge1$ and $x\ge y\ge1$. Then
	\begin{align}\label{eq:mobius-event-convolution}
		E_0[t]&=\sum_{d\mid t}\mu(d)G_0(t/d),\\
		E_1[t]&=\sum_{d\mid t}\mu(d)d\,G_1(t/d),\nonumber\\
		E_2[t]&=\sum_{d\mid t}\mu(d)d^2G_2(t/d).\nonumber
	\end{align}
	The convolution is evaluated by looping over $d$ and then over multiples $t=dL\le N$. Its cost is $\sum_{d\le N}O(N/d)=O(N\log N)$, so the main remaining task is the construction of the three arrays $G_i$.
	
	\subsection{Constructing the coefficient arrays}
	
	Let $B=\lfloor\sqrt N\rfloor$. Since $L=ax+by\le N$ implies $ax\le N$, every quadruple satisfies $a\le B$ or $x\le B$. We use the disjoint cover
	\begin{equation}\label{eq:all-values-disjoint-cover}
		\{x\le B\}\ \sqcup\ \{a\le B,\ x>B\}.
	\end{equation}
	
	In the part $a\le B$, fix $a>b$ and write $x=y+s$ with $s\ge0$. Then
	\[
	L=(a+b)y+as,
	\qquad
	K=(a+b)y+bs.
	\]
	The disjoint condition is $x=y+s>B$. Hence for fixed $a,b,y$ the parameter $s$ runs over the interval
	\[
	\max(0,B+1-y)\le s\le \left\lfloor\frac{N-(a+b)y}{a}\right\rfloor,
	\]
	with the interval omitted if the upper bound is smaller than the lower bound. Along this interval $L$ is an arithmetic progression with step $a$, and the updates to $G_0,G_1,G_2$ are polynomials of degree at most two in $s$. The possible singleton $s=0$ has multiplier $2$, while the range $s\ge1$ has multiplier $4$.
	
	In the part $x\le B$, fix $x>y$ and write $a=b+r$ with $r\ge1$. Then
	\[
	L=(x+y)b+xr,
	\qquad
	K=(x+y)b+yr.
	\]
	For fixed $x,y,r$, the parameter $b$ runs over
	\[
	1\le b\le \left\lfloor\frac{N-xr}{x+y}\right\rfloor,
	\]
	again with empty intervals ignored. Along this interval $L$ is an arithmetic progression with step $x+y$, and the three coefficient updates are polynomials of degree at most two in $b$. Finally, for diagonal sides $x=y=t$, write $s=a+b$. Since $L=K=ts$ and the number of pairs $a>b\ge1$ with $a+b=s$ is $\lfloor(s-1)/2\rfloor$, we add, for $ts\le N$,
	\begin{align}\label{eq:all-values-diagonal-side}
		G_0[ts]&\mathrel{+}=2\Floor{\frac{s-1}{2}},\\
		G_1[ts]&\mathrel{+}=4ts\Floor{\frac{s-1}{2}},\nonumber\\
		G_2[ts]&\mathrel{+}=2(ts)^2\Floor{\frac{s-1}{2}}.\nonumber
	\end{align}
	
	All non-diagonal updates have the generic form
	\begin{equation}\label{eq:ap-update}
		G[L_0+m\ell]\mathrel{+}=c_0+c_1\ell+c_2\ell^2,
		\qquad \ell_0\le\ell\le\ell_1.
	\end{equation}
	For a fixed step $m\le2B$ and residue class modulo $m$, rewrite the right-hand side as a quadratic polynomial in the quotient index of the progression. Three ordinary difference arrays add such a polynomial on an interval in $O(1)$ time. The construction processes the steps one at a time: while step $m$ is active, the buffers store only the quotient-index difference arrays for the residue classes modulo $m$. Their combined length is $O(N)$, because every integer $L\le N$ belongs to exactly one residue class for this fixed step. After all updates with step $m$ have been inserted, a single flush over the residue classes materializes their contribution to the global arrays $G_0,G_1,G_2$; the buffers are then cleared and reused for the next step. Thus the $O(N)$ flush cost is paid separately for each $m$, but the buffer memory is not multiplied by the number of steps.
	
	\begin{lemma}[coefficient-array construction]\label{lem:coefficient-array-construction}
		The arrays $G_0,G_1,G_2$ for all $L\le N$ can be constructed in $O(N^{3/2})$ arithmetic operations and $O(N)$ memory.
	\end{lemma}
	\begin{proof}
		The cover \eqref{eq:all-values-disjoint-cover} is disjoint and exhaustive, so no tuple is missed or counted twice. In the part $a\le B$, for fixed $a,b$ the number of possible $y$ is $O(N/(a+b))$, so the number of arithmetic-progression updates is
		\[
		\sum_{a\le B}\sum_{b<a}O\!\left(\frac{N}{a+b}\right)=O(NB),
		\]
		because $\sum_{b<a}(a+b)^{-1}=O(1)$ for each $a$. In the non-diagonal part $x\le B$, for fixed $x,y$ the number of possible $r$ is $O(N/x)$, so the number of updates is
		\[
		\sum_{x\le B}\sum_{y<x}O\!\left(\frac{N}{x}\right)=O(NB).
		\]
		The diagonal-side updates contribute $\sum_{t\le B}O(N/t)=O(N\log N)$ more. For each step $m\le2B$, the total length of all residue-class buffers is $O(N)$ and flushing them once costs $O(N)$; over all $O(B)$ possible steps this gives $O(NB)$ flush time. Since the steps are processed sequentially, the same buffers are reused after each flush. Therefore the working memory consists of the global arrays $G_i$ plus the buffers for a single step, all of total size $O(N)$, rather than $O(NB)$. With $B=\lfloor\sqrt N\rfloor$, the total time is $O(N^{3/2})$.
	\end{proof}
	
	\begin{theorem}[all values]\label{thm:all-values}
		The whole table $F(1),F(2),\ldots,F(N)$ can be computed exactly in
		$
		O(N^{3/2})
		$
		arithmetic operations and $O(N)$ memory.
	\end{theorem}
	\begin{proof}
		By \cref{lem:coefficient-array-construction}, the divisor-free coefficient arrays are built in $O(N^{3/2})$ time. The M\"obius convolution \eqref{eq:mobius-event-convolution} costs $O(N\log N)$, which is absorbed by $O(N^{3/2})$ for $N\ge2$. Prefixing the three event arrays and applying \eqref{eq:all-values-recover} for all $n\le N$ is linear. Adding the closed form for $F_0(n)$ for every $n$ is also linear.
	\end{proof}
	
	\section{Asymptotics}\label{sec:asymptotics}
	
	In this section we state the final two-term asymptotic expansion and explain the structure of its proof. This part of the paper is somewhat more independent than the algorithmic sections, but we include it here because it gives a stringent large-$n$ consistency check for the exact values computed by the algorithms. The full derivation is deferred to \cref{app:second-term}.
	
	The directions collected in $F_0(n)$ contribute only at order $n^4$, so the logarithmic term comes entirely from the primitive directions with $u>v\ge 1$. The key observation is that the remaining part admits an inclusion--exclusion decomposition associated with the covering
	\[
	\{\text{admissible quadruples}\}=\{u,v\le \sqrt n\}\cup\{a,b\le \sqrt n\},
	\]
	whose precise form is proved later as the covering identity in \cref{lem:asym-cover}. Writing $S_1(n)$ for the contribution of the region $u,v\le \sqrt n$, $S_2(n)$ for the contribution of the region $a,b\le \sqrt n$, and $S_{12}(n)$ for their overlap, one therefore has
	$
	F(n)-F_0(n)=S_1(n)+S_2(n)-S_{12}(n).
	$
	
	The terms $S_1(n)$ and $S_2(n)$ are symmetric: after rescaling, each is governed by the same integral kernel and therefore contributes the same coefficient to the $n^4\log n$ term. The overlap term $S_{12}(n)$ is counted in both pieces and must be subtracted once by inclusion--exclusion; it contributes only at order $n^4$. Thus the logarithmic main term is obtained by computing one of the two symmetric pieces and doubling its contribution, while the constant-order term comes from keeping track of all three pieces together with $F_0(n)$.
	
	\begin{theorem}[two-term asymptotic expansion]\label{thm:main-asymptotic}
		Let $F(n)$ be the total number of lattice rectangles contained in $[0,n)\times[0,n)$. Then
		$
		F(n)=A\,n^4\log n+B\,n^4+o(n^4),
		$
		where
		\[
		A=\frac{4\log 2-1}{\pi^2}, \ \ 
		B=
		-\frac{4\log 2-1}{6}\frac{\zeta'(2)}{\zeta(2)^2}
		+\frac{24(4\log 2-1)\gamma+72\log^2 2-76\log 2+1}{12\pi^2}
		-\frac14.
		\]
	\end{theorem}
	
	The proof, including the decomposition into $S_1$, $S_2$, and $S_{12}$ and the evaluation of the corresponding constants to order $o(n^4)$, is given in \cref{app:second-term}.
	A numerical validation of the constants using exact values at powers of two is reported in \cref{sec:large-values}.
	
	\section{Experiments}\label{sec:experiments}
	
	All benchmark implementations, reference Python versions, the CUDA code used for the largest one-value computations, and the compressed all-values data are available at \url{https://github.com/flykiller/lattice-rectangles}. The one-value algorithms were benchmarked in single-threaded C++ on one core of an Intel i7-13700 at 5.2\,GHz, compiled with \texttt{clang++} 20.1.8 and \texttt{-O3 -march=native}, using \texttt{std::int128\_t} throughout the tested range. The plots in \cref{fig:timings} report wall-clock seconds for \(n=2^{13},\dots,2^{30}\); each point is the median of repeated runs, with relative variation below \(0.5\%\).
	
	\subsection{Comparison of one-value algorithms}
	
	\Cref{fig:timings} compares all five implemented one-value algorithms: the quadratic primitive-direction sweep, the square-root decomposition, the cubic-root moment reduction, the ten-moment weighted floor-sum reduction, and the divisor-layer algorithm of \cref{sec:one-value}. The log--log plots are nearly linear with the expected slopes. The right panel normalizes each curve by its proved complexity, namely by \(n^2\), \(n^{3/2}\log n\), \(n^{4/3}\log n\), \(n\log^3 n\), and \(n\log^2 n\), respectively.
	
	The comparison shows the practical effect of the successive reductions. The square-root and cubic-root reorganizations already reduce the growth rate substantially; the weighted floor-sum kernel improves the one-value computation further; and the divisor-layer method gives the best asymptotic bound among the tested one-value algorithms. Its curve is included in the same plot because it is computing the same quantity \(F(n)\), with the M\"obius divisor layers replacing the older outer organization.
	
	\begin{figure}[t]
		\centering
		\includegraphics[width=0.9\textwidth]{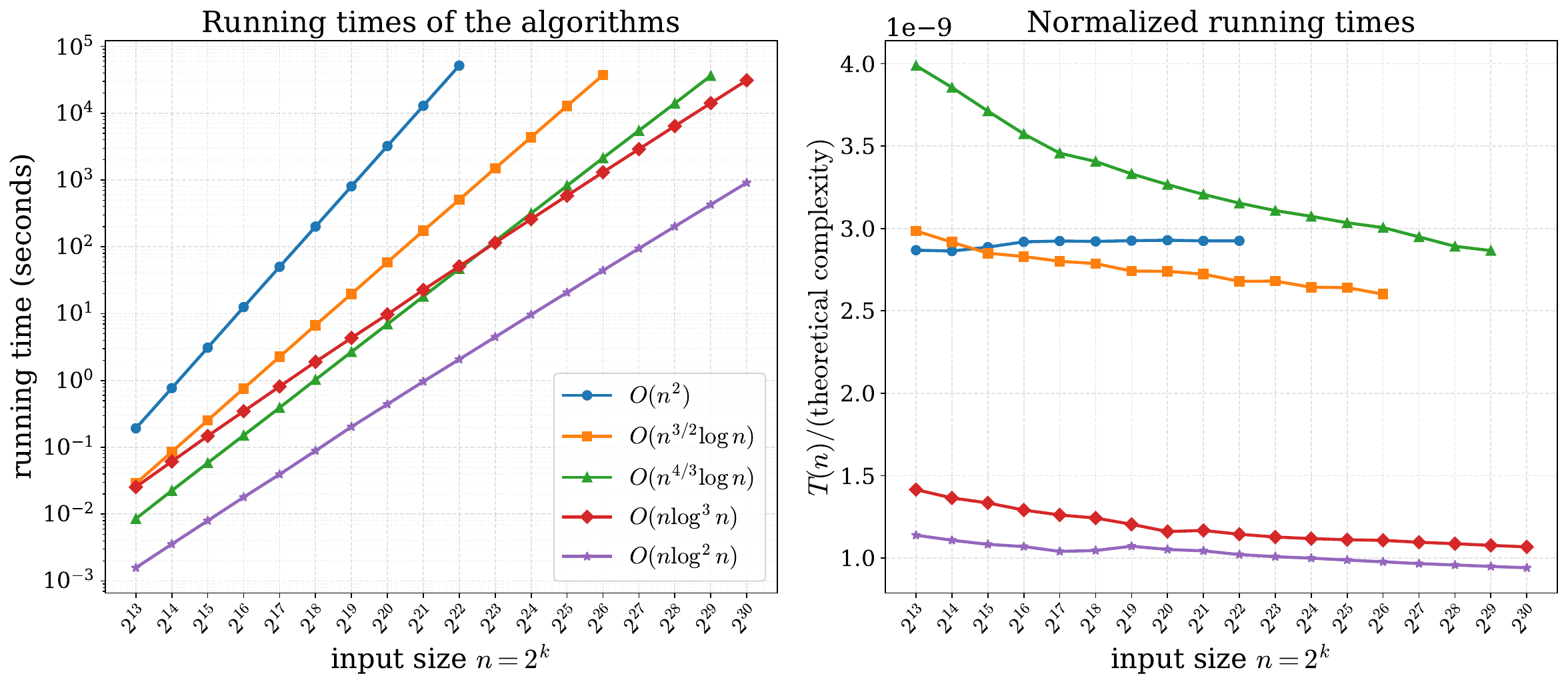}
		\caption{Wall-clock running times of the five tested one-value algorithms on inputs $n=2^{13},\dots,2^{30}$ (left) and after normalization by their respective theoretical complexities (right).}
		\label{fig:timings}
	\end{figure}
	
	\subsection{All-values experiment}
	
	We also benchmarked the all-values algorithm of \cref{sec:all-values}, which computes the whole table \(F(1),\ldots,F(N)\) in one run. The implementation was tested for \(N=2^5,2^6,\ldots,2^{24}\). In addition, we used the same \(O(N^{3/2})\) implementation to compute all exact values up to \(N=10^8\). The full table is too large to include here, but the results are available in compressed form in the GitHub repository mentioned above. In \cref{fig:all-values-timings}, the left panel gives the total wall-clock time, while the right panel plots the normalized quantity \(10^6T(N)/N^{3/2}\).
	
	The normalized curve is not perfectly flat: after the smallest inputs, it gradually increases in the larger range of the experiment. This does not contradict the \(O(N^{3/2})\) operation count. The all-values implementation maintains several arrays of length \(N\) and performs repeated strided updates and flushes through these arrays. For small \(N\), most of this working set stays in cache; as \(N\) grows, the coefficient arrays, event arrays, and temporary buffers exceed the faster cache levels, and the computation becomes increasingly limited by cache misses, TLB pressure, and memory bandwidth. Thus the observed growth of \(T(N)/N^{3/2}\) is best understood as a memory-hierarchy effect rather than as a change in the arithmetic complexity.
	
	This algorithm is also much less convenient to parallelize than the one-value divisor-layer computation. Its updates are not independent outputs: many iterations write into the same coefficient arrays, event arrays, or residue-class buffers. A parallel implementation must therefore either keep large private copies of these arrays for different workers and merge them afterwards, which greatly increases memory consumption, or let workers share the arrays and pay for synchronization or atomic writes, which creates many write conflicts. For this reason the experiment above is intended as a single-threaded baseline for the \(O(N^{3/2})\) all-values method.
	
	\begin{figure}[t]
		\centering
		\includegraphics[width=0.9\textwidth]{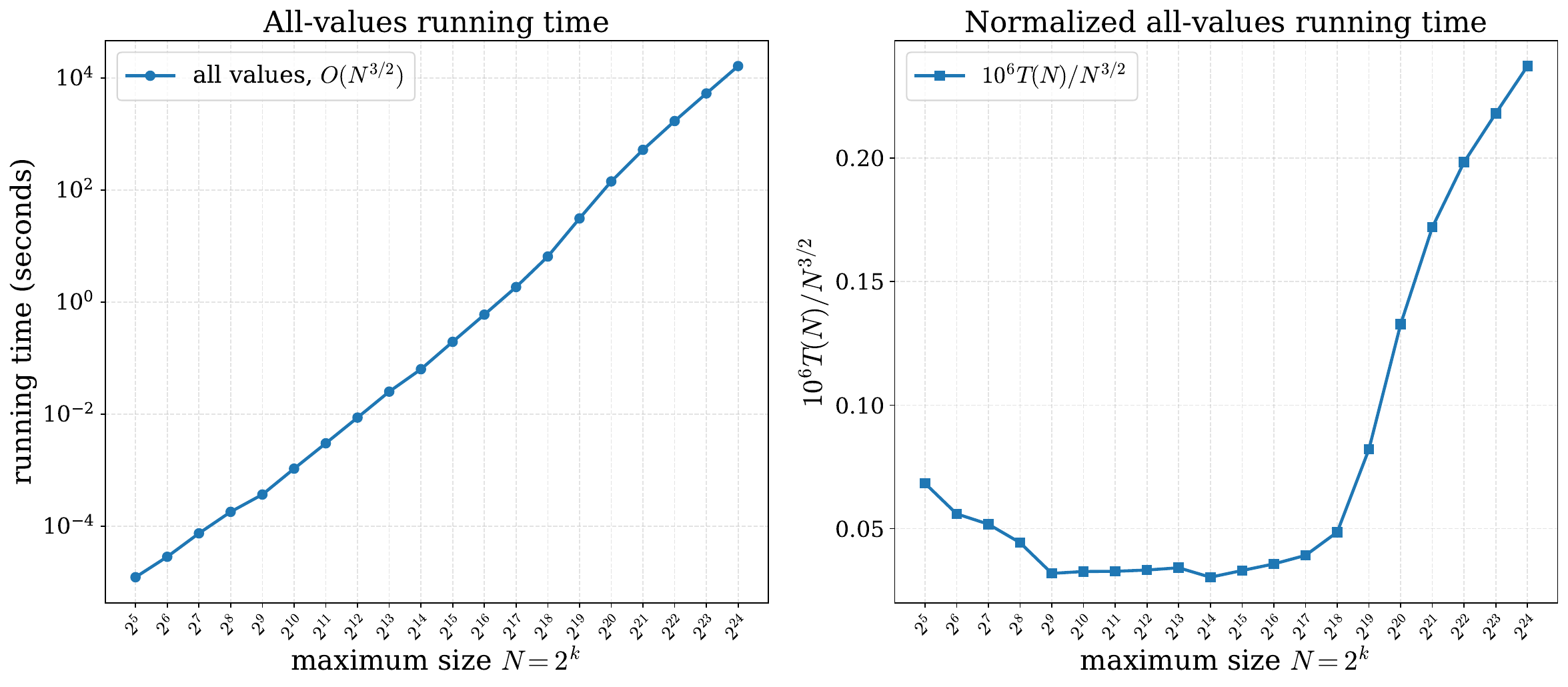}
		\caption{All-values experiment for \(N=2^5,\dots,2^{24}\): total running time (left) and normalized running time \(10^6T(N)/N^{3/2}\) (right).}
		\label{fig:all-values-timings}
	\end{figure}
	
	\subsection{Large exact values and asymptotic check}\label{sec:large-values}
	
	For convenience we record exact values of $F(2^k)$ for $1\le k\le40$. Smaller entries were cross-checked against slower exact methods whenever feasible. The largest entries, up to $2^{40}$, were obtained with a separate CUDA implementation of the divisor-layer one-value formula of \cref{sec:one-value}. At a high level, the CUDA code uses a segmented GPU sieve to build M\"obius-summed divisor layers, splits each layer into independent arithmetic-progression work items, processes these work items with persistent CUDA blocks, accumulates exact partial sums in a fixed-width 192-bit representation, and performs a final reduction. These values also serve as input for the asymptotic comparison below.

	\begin{table}[ht]
		\centering
		{
			\renewcommand{\arraystretch}{1.00}
			\setlength{\tabcolsep}{4pt}
			\resizebox{\textwidth}{!}{%
				\begin{tabular}{@{}rr|rr@{}}
					\toprule
					$k$ & $F(2^k)$ & $k$ & $F(2^k)$\\
					\midrule
					$1$ & $1$ & $21$ & $48931868439876126051425552$ \\
					$2$ & $44$ & $22$ & $821437615651793675198669752$ \\
					$3$ & $1192$ & $23$ & $13759445380252558103053449112$ \\
					$4$ & $27128$ & $24$ & $230014222561387209679445816240$ \\
					$5$ & $564120$ & $25$ & $3838037104619867210112196814232$ \\
					$6$ & $11114080$ & $26$ & $63933546372113490066412405897360$ \\
					$7$ & $211224480$ & $27$ & $1063335985124949941305863686097296$ \\
					$8$ & $3914221216$ & $28$ & $17659763652737469299382592232330696$ \\
					$9$ & $71182606216$ & $29$ & $292898424695610564494215857912343064$ \\
					$10$ & $1275797150128$ & $30$ & $4851850095158746095561485451592336296$ \\
					$11$ & $22602804487208$ & $31$ & $80277206323003614389748671287223855080$ \\
					$12$ & $396685572297544$ & $32$ & $1326796977975476403092689286862986516504$ \\
					$13$ & $6907621416632376$ & $33$ & $21906538476526319541299023010218991588136$ \\
					$14$ & $119492377263166968$ & $34$ & $361349204887120272089523042249821840571528$ \\
					$15$ & $2055404973525169560$ & $35$ & $5955100706397110811260922659812491131662432$ \\
					$16$ & $35182910663019639384$ & $36$ & $98057826153604756744005601368029402514221504$ \\
					$17$ & $599669468453524178752$ & $37$ & $1613344656077691850026984888873116366804460232$ \\
					$18$ & $10182597857710132553464$ & $38$ & $26524225499163321460061315970545176007812869616$ \\
					$19$ & $172327747508964813792096$ & $39$ & $435758984017337173124103405065600778830350047408$ \\
					$20$ & $2907742868855598433202344$ & $40$ & $7154085760768979246024995359851578213153827420872$ \\
					\bottomrule
				\end{tabular}%
			}
		}
		\caption{Exact values of $F(2^k)$ for $1\le k\le40$.}
		\label{tab:powers-of-two-values}
	\end{table}
	
	\medskip
	\noindent\textbf{Two-term fit.}
	With the numerical constants \(A=(4\log 2-1)/\pi^2\) and \(B\approx -0.084567061533\), evaluated in the appendix, it is natural to inspect
	$
	\frac{F(n)-A n^4\log n}{n^4}.
	$
	For the computed values \(n=2^k\), this quantity appears to stabilize rapidly near the predicted constant \(B\). Thus the data are fully consistent with the two-term asymptotic \(F(n)=A n^4\log n+B n^4+o(n^4)\), and provide a simple numerical check.

	\section{Summary and future work}
	
	We obtained exact one-value algorithms of complexity $O(n^2)$ for the classical primitive-direction sweep, $O(n^{3/2}\log n)$ for the square-root decomposition, $O(n^{4/3}\log n)$ for the cubic-root decomposition with floor moments, $O(n\log^3 n)$ for the ten-moment weighted floor-sum reduction, and $O(n\log^2 n)$ for the divisor-layer square-root-cover algorithm. We also obtained an all-values algorithm computing $F(1),\ldots,F(N)$ in $O(N^{3/2})$ time and $O(N)$ memory. The key structural feature of the fastest stages is the collapse of the geometric summation to one-dimensional arithmetic kernels or exact-threshold coefficient arrays: a six-state floor-moment kernel in the $O(n^{4/3}\log n)$ and $O(n\log^2 n)$ algorithms, a ten-state Euclidean kernel in the $O(n\log^3 n)$ reduction, and arithmetic-progression event updates in the all-values algorithm. \Cref{app:engineering} records implementation-level simplifications that preserve the asymptotic complexity while improving the practical running time by more than a factor of three.
	
	Natural directions include further reducing the one-value complexity below $O(n\log^2 n)$, sharpening the remainder beyond \(o(n^4)\), extending the methods to other lattice objects and higher dimensions, and improving the all-values construction below the square-root barrier. A genuinely sublinear algorithm for a single value $F(n)$ would likely require a different idea: even in the divisor-layer formulation there are $\Theta(n)$ natural outer scales to account for, so an \(\Omega(n)\)-type barrier appears heuristically unavoidable unless one finds additional global cancellations or a new transform-level description of the count.
	
	\newpage
	
	\appendix

	\section{Algorithmic appendix: expanded pseudocode}\label{app:pseudocode}
	
	This appendix records more explicit procedural versions of the algorithmic stages. The aim is not to specify low-level implementation details, but to make the control flow of each algorithm and the interfaces between the geometric summation and the one-dimensional kernels transparent. Throughout this appendix, the three directions $(0,1)$, $(1,0)$, and $(1,1)$ are excluded from the loops and are handled separately through the closed-form term $F_0(n)$.
	
	\setlength{\abovedisplayskip}{5pt}
	\setlength{\belowdisplayskip}{5pt}
	\algrenewcommand\algorithmicindent{1.1em}
	\captionsetup[algorithm]{font=small,skip=4pt}
	\makeatletter
	\renewcommand\ALG@name{Algorithm}
	\makeatother
	
	\begin{algorithm}[H]
		\small
		\caption{Classical quadratic baseline.}
		\begin{algorithmic}[1]
			\State Initialize $ans\gets 0$.
			\For{each primitive $(u,v)$ with $u>v>0$}
			\State determine $y_{\min}\le y\le y_{\max}$ from \cref{eq:basic-constraints}
			\For{$y=y_{\min},\dots,y_{\max}$}
			\State determine the induced interval $x_{\min}(y)\le x\le x_{\max}(y)$
			\State if nonempty, add the closed polynomial sum of $\mult(x,y)$ times \eqref{eq:placements} over $x$
			\EndFor
			\EndFor
			\State \Return $F_0(n)+ans$.
		\end{algorithmic}
	\end{algorithm}
	
	\begin{algorithm}[H]
		\small
		\caption{Square-root decomposition.}
		\begin{algorithmic}[1]
			\State Set $B\gets \Floor{\sqrt n}$ and $ans\gets 0$.
			\For{each primitive $(u,v)$ with $u\le B$ and $u>v>0$}
			\State evaluate $c_{u,v}(n)$ by the baseline routine and add $c_{u,v}(n)$
			\EndFor
			\For{each $(x,y)$ with $x\ge y\ge 1$ and $x\le \Floor{n/B}$}
			\State pass to the dual $(x,y)$-parametrization and determine $L\le v\le U$
			\State accumulate the primitive contribution on $[L,U]$ by coprime prefix sums, including the multiplier $\mult(x,y)$
			\EndFor
			\State \Return $F_0(n)+ans$.
		\end{algorithmic}
	\end{algorithm}
	
	\begin{algorithm}[H]
		\small
		\caption{Cubic-root decomposition.}
		\begin{algorithmic}[1]
			\State Set $B\gets \Floor{n^{2/3}}$ and $ans\gets 0$.
			\For{each primitive $(u,v)$ with $u\le B$ and $u>v>0$}
			\State rewrite $c_{u,v}(n)$ as a fixed linear combination of six moments
			\State evaluate the required values $\mathbf H(\cdot;\cdot,\cdot,\cdot)$ by \cref{app:six-moments} and add $c_{u,v}(n)$
			\EndFor
			\For{$u=B+1,\dots,n$}
			\State set $x_{\max}\gets \left\lfloor n/u\right\rfloor$
			\For{each $(x,y)$ with $1\le y\le x\le x_{\max}$}
			\State set $m\gets 2$ if $x=y$, and $m\gets 4$ otherwise
			\State set $v_{\max}\gets \min\!\left(u-1,\left\lfloor\frac{n-xu}{y}\right\rfloor\right)$
			\State if $v_{\max}<1$, continue to the next pair $(x,y)$
			\State evaluate $(C_0,C_1,C_2)$ on $1\le v\le v_{\max}$ by coprime prefix sums
			\State add $m\bigl((n-xu)(n-yu)C_0-(x(n-xu)+y(n-yu))C_1+xyC_2\bigr)$
			\EndFor
			\EndFor
			\State \Return $F_0(n)+ans$.
		\end{algorithmic}
	\end{algorithm}
	
	\begin{algorithm}[H]
		\small
		\caption{Ten-moment weighted-floor-sum algorithm.}
		\begin{algorithmic}[1]
			\State Precompute the squarefree divisor lists
			$D(u)=\{(d,\mu(d)): d\mid u,\ d\text{ squarefree}\}$ for all $u\le n$.
			\State Initialize $ans\gets 0$.
			\For{$u=2,\dots,n$ and each $y$ with $1\le y\le \lfloor n/u\rfloor$}
			\For{each $(d,\mu(d))\in D(u)$ with $d<u$ and $y(u+d)\le n$}
			\State form the M\"obius-expanded quantity $R_{u,y,d}$
			\State split $R_{u,y,d}$ into its polynomial and weighted floor-sum parts
			\State evaluate the weighted queries by \cref{app:ten-moments}
			\State add $\mu(d)R_{u,y,d}$ to $ans$
			\EndFor
			\EndFor
			\State \Return $F_0(n)+ans$.
		\end{algorithmic}
	\end{algorithm}
	
	\begin{algorithm}[H]
		\small
		\caption{Divisor-layer one-value algorithm.}
		\begin{algorithmic}[1]
			\State Precompute $\mu(d)$ for $d\le n$.
			\State Initialize $ans\gets 0$.
			\For{$d=1,\dots,n$ with $\mu(d)\ne0$}
			\State set $N_d\gets\lfloor n/d\rfloor$, $B\gets\lfloor\sqrt{N_d}\rfloor$, and $S\gets0$
			\Comment{use $\{x\le B\}\sqcup\{a\le B,\ x>B\}$}
			\For{$1\le q<p\le B$}
			\State evaluate the moments $M_{ij}(p,q;N_d)$ and $M^{\cap}_{ij}(p,q;N_d,B)$ for $0\le i+j\le2$ by \cref{app:six-moments}
			\State form $\Phi_{p,q}^{(d)}$, $\Delta_{p,q}^{(d)}$, $\Phi_{p,q}^{(d),\cap}$, and $\Delta_{p,q}^{(d),\cap}$
			\State add $8\Phi_{p,q}^{(d)}-6\Delta_{p,q}^{(d)}-4\Phi_{p,q}^{(d),\cap}+2\Delta_{p,q}^{(d),\cap}$ to $S$
			\EndFor
			\For{$t=1,\dots,B$}
			\State add the closed diagonal-side sum $C_t^{(d)}$ to $S$
			\EndFor
			\State add $\mu(d)S$ to $ans$
			\EndFor
			\State \Return $F_0(n)+ans$.
		\end{algorithmic}
	\end{algorithm}
	
	\begin{algorithm}[H]
		\small
		\caption{All-values event-array algorithm.}
		\begin{algorithmic}[1]
			\State Set $B\gets\lfloor\sqrt N\rfloor$ and initialize $G_0,G_1,G_2$ to zero.
			\State Insert the arithmetic-progression updates from the disjoint cover $\{x\le B\}\sqcup\{a\le B,\ x>B\}$.
			\State Flush the step buffers to materialize all entries of $G_0,G_1,G_2$.
			\State Precompute $\mu(d)$ for $d\le N$ and initialize $E_0,E_1,E_2$ to zero.
			\For{$d=1,\dots,N$ and all multiples $t=dL\le N$}
			\State add $\mu(d)G_0[L]$, $\mu(d)dG_1[L]$, and $\mu(d)d^2G_2[L]$ to $E_0[t]$, $E_1[t]$, and $E_2[t]$
			\EndFor
			\State Prefix the three event arrays to obtain $P_0,P_1,P_2$.
			\For{$n=1,\dots,N$}
			\State output $F(n)=F_0(n)+n^2P_0(n)-nP_1(n)+P_2(n)$.
			\EndFor
		\end{algorithmic}
	\end{algorithm}

	The squarefree divisor lists used in the ten-moment algorithm can be generated once by a sieve for the smallest prime factor together with the recursive construction of all squarefree divisors of each $u$. Their total size up to $n$ is $\sum_{u\le n}2^{\omega(u)}=O(n\log n)$, so this preprocessing does not affect the $O(n\log^3 n)$ complexity bound.
	
	\section{A six-moment weighted kernel}\label{app:six-moments}
	
	This appendix records the floor-moment kernel used in the small-direction part of the $O(n^{4/3}\log n)$ algorithm. In the main text, each contribution $c_{u,v}(n)$ with $2\le u\le B$ and $1\le v<u$ is reduced to a fixed linear combination of floor sums of the form $\sum x^p\lfloor(ax+b)/m\rfloor^q$ with $p+q\le 3$; the six-moment family below is exactly the closure of those sums under the Euclidean recursion.
	
	\subsection{Definition of the family}
	
	For integers $n\ge 0$, $m\ge 1$, and $a,b\ge 0$, define
	\[
	\mathcal H_{p,q}(n;m,a,b)
	:=
	\sum_{x=0}^{n-1} x^p
	\Floor{\frac{ax+b}{m}}^q,
	\qquad q\ge 1,\quad p+q\le 3.
	\]
	Equivalently, we work with the six quantities
	\[
	\mathbf H(n;m,a,b)
	:=
	(\mathcal H_{0,1},\ \mathcal H_{1,1},\ \mathcal H_{2,1},\ \mathcal H_{0,2},\ \mathcal H_{1,2},\ \mathcal H_{0,3}).
	\]
	We also use the power sums
	\[
	P_r(n):=\sum_{x=0}^{n-1} x^r,
	\qquad r=0,1,2,3,
	\]
	namely
	\[
	P_0(n)=n,\qquad
	P_1(n)=\frac{n(n-1)}{2},\qquad
	P_2(n)=\frac{n(n-1)(2n-1)}{6},\qquad
	P_3(n)=P_1(n)^2.
	\]
	
	\subsection{Base case}
	
	If $a=0$, then $\Floor{(ax+b)/m}=c:=\Floor{b/m}$ is constant, and therefore
	\[
	\mathbf H(n;m,0,b)
	=
	\bigl(cP_0,\ cP_1,\ cP_2,\ c^2P_0,\ c^2P_1,\ c^3P_0\bigr),
	\]
	where $P_r=P_r(n)$.
	
	\subsection{Affine step}
	
	Assume $a\ge m$ or $b\ge m$. Write
	\[
	a=Am+a',\qquad b=Bm+b',
	\qquad 0\le a',b'<m.
	\]
	Then
	\[
	\Floor{\frac{ax+b}{m}}
	=
	Ax+B+\Floor{\frac{a'x+b'}{m}}.
	\]
	Hence the six moments for $(m,a,b)$ are explicit linear combinations of the six moments for $(m,a',b')$ and of the power sums.
	
	Let
	\[
	\mathbf B:=\mathbf H(n;m,a',b')=(b_{01},b_{11},b_{21},b_{02},b_{12},b_{03}),
	\]
	and abbreviate
	\[
	P_r:=P_r(n).
	\]
	Then the affine reduction gives
	\begin{align*}
		\mathcal H_{0,1}(n;m,a,b)
		&=
		b_{01}+AP_1+BP_0,\\
		\mathcal H_{1,1}(n;m,a,b)
		&=
		b_{11}+AP_2+BP_1,\\
		\mathcal H_{2,1}(n;m,a,b)
		&=
		b_{21}+AP_3+BP_2,\\[0.5ex]
		\mathcal H_{0,2}(n;m,a,b)
		&=
		b_{02}+2Ab_{11}+2Bb_{01}+A^2P_2+2ABP_1+B^2P_0,\\
		\mathcal H_{1,2}(n;m,a,b)
		&=
		b_{12}+2Ab_{21}+2Bb_{11}+A^2P_3+2ABP_2+B^2P_1,\\[0.5ex]
		\mathcal H_{0,3}(n;m,a,b)
		&=
		b_{03}
		+3Ab_{12}+3Bb_{02}
		+3A^2b_{21}+6ABb_{11}+3B^2b_{01}\\
		&\qquad
		+A^3P_3+3A^2BP_2+3AB^2P_1+B^3P_0.
	\end{align*}
	
	\subsection{Reciprocal step}
	
	Now assume
	\[
	0\le a,b<m.
	\]
	Set
	\[
	Y:=\Floor{\frac{a(n-1)+b}{m}}.
	\]
	If $Y=0$, then all six moments vanish.
	
	Otherwise define
	\[
	g(t):=\Floor{\frac{mt+(m-b-1)}{a}},
	\qquad 0\le t\le Y-1.
	\]
	As usual, the reciprocal transformation exchanges the graph of
	\[
	f(x):=\Floor{\frac{ax+b}{m}}
	\]
	with the complementary staircase determined by $g$, and therefore expresses the moments for $(n;m,a,b)$ through the moments for $(Y;a,m,m-b-1)$.
	
	Write
	\[
	\mathbf G
	=
	\mathbf H(Y;a,m,m-b-1)
	=
	(g_{01},g_{11},g_{21},g_{02},g_{12},g_{03}),
	\]
	and let
	\[
	P_r:=P_r(n),\qquad Q_r:=P_r(Y).
	\]
	Then
	\begin{align*}
		\mathcal H_{0,1}(n;m,a,b)
		&=
		P_0Y-Q_0-g_{01},\\
		\mathcal H_{1,1}(n;m,a,b)
		&=
		\frac{2P_1Y-g_{01}-g_{02}}{2},\\
		\mathcal H_{2,1}(n;m,a,b)
		&=
		\frac{6P_2Y-g_{01}-3g_{02}-2g_{03}}{6},\\[0.5ex]
		\mathcal H_{0,2}(n;m,a,b)
		&=
		P_0Y^2-Q_0-g_{01}-2Q_1-2g_{11},\\
		\mathcal H_{1,2}(n;m,a,b)
		&=
		\frac{2P_1Y^2-g_{01}-g_{02}-2g_{11}-2g_{12}}{2},\\[0.5ex]
		\mathcal H_{0,3}(n;m,a,b)
		&=
		P_0Y^3-Q_0-g_{01}-3Q_1-3g_{11}-3Q_2-3g_{21}.
	\end{align*}
	
	\begin{lemma}\label{lem:six-moment-kernel}
		The six moments
		\[
		\mathcal H_{0,1},\ \mathcal H_{1,1},\ \mathcal H_{2,1},\ \mathcal H_{0,2},\ \mathcal H_{1,2},\ \mathcal H_{0,3}
		\]
		admit a recursive evaluation in $O(\log m)$ arithmetic operations.
	\end{lemma}
	
	\begin{proof}
		The affine step reduces $(a,b)$ modulo $m$. The reciprocal step replaces $(n;m,a,b)$ by
		\[
		(Y;a,m,m-b-1),
		\qquad
		Y=\Floor{\frac{a(n-1)+b}{m}},
		\]
		so the first modulus strictly decreases from $m$ to $a<m$. Hence the recursion has Euclidean depth $O(\log m)$. Each step performs only $O(1)$ arithmetic operations on the six stored moments and the power sums.
	\end{proof}
	
	\subsection{Complexity bound}
	
	Each recursive call either reduces $a$ and $b$ modulo $m$ (the affine step) or swaps the roles of $(a,m)$ with strictly smaller Euclidean parameters (the reciprocal step). Hence the recursion depth is $O(\log m)=O(\log n)$. Every step performs only $O(1)$ arithmetic operations on a constant-size family of moments.
	
	\section{A ten-moment weighted kernel}\label{app:ten-moments}
	
	This appendix records a compact floor-moment kernel sufficient for the weighted reduction of \cref{eq:main-floor-family}. In the main text, each outer term $R_{u,y,d}$ is reduced to a constant number of weighted floor sums whose polynomial weights have total degree at most $4$; the ten moments below are the minimal closed family we use to evaluate those queries recursively.
	
	\subsection{Definition of the family}
	
	For integers $n\ge 0$, $m\ge 1$, and $a,b\ge 0$, define
	\[
	\mathcal H_{p,q}(n;m,a,b)
	:=
	\sum_{x=0}^{n-1} x^p
	\Floor{\frac{ax+b}{m}}^q,
	\qquad q\ge 1,\quad p+q\le 4.
	\]
	Equivalently, we work with the ten quantities
	\[
	\mathbf H(n;m,a,b):= (\mathcal H_{0,1},\ \mathcal H_{1,1},\ \mathcal H_{2,1},\ \mathcal H_{3,1},\ 
	\mathcal H_{0,2},\ \mathcal H_{1,2},\ \mathcal H_{2,2},\ 
	\mathcal H_{0,3},\ \mathcal H_{1,3},\ 
	\mathcal H_{0,4}).
	\]
	We also use the power sums
	\[
	P_r(n):=\sum_{x=0}^{n-1} x^r,
	\qquad r=0,1,2,3,4,
	\]
	namely
	\[
	P_0(n)=n,\qquad
	P_1(n)=\frac{n(n-1)}{2},\qquad
	P_2(n)=\frac{n(n-1)(2n-1)}{6},
	\]
	\[
	P_3(n)=\left(\frac{n(n-1)}{2}\right)^2,
	\qquad
	P_4(n)=\frac{n(n-1)(2n-1)(3n^2-3n-1)}{30}.
	\]

	\subsection{Base case}
	If $a=0$, then $\Floor{(ax+b)/m}=c:=\Floor{b/m}$ is constant, and therefore
	\[
	\mathbf H(n;m,0,b)
	=
	\bigl(
	cP_0,\ cP_1,\ cP_2,\ cP_3,\ 
	c^2P_0,\ c^2P_1,\ c^2P_2,\ 
	c^3P_0,\ c^3P_1,\ 
	c^4P_0
	\bigr),
	\]
	where $P_r=P_r(n)$.
	
	\subsection{Affine step}
	
	Assume $a\ge m$ or $b\ge m$.  Write
	\[
	a=Am+a',\qquad b=Bm+b',
	\qquad 0\le a',b'<m.
	\]
	Then
	\[
	\Floor{\frac{ax+b}{m}}
	=
	Ax+B+\Floor{\frac{a'x+b'}{m}}.
	\]
	Hence the ten moments for $(m,a,b)$ are explicit linear combinations of the ten moments for $(m,a',b')$ and of the power sums.
	
	Let
	\[
	\mathbf B:=\mathbf H(n;m,a',b')
	=
	(b_{01},b_{11},b_{21},b_{31},b_{02},b_{12},b_{22},b_{03},b_{13},b_{04}),
	\]
	and abbreviate
	\[
	P_r:=P_r(n).
	\]
	Then the affine reduction gives
	\begin{align*}
		\mathcal H_{0,1}(n;m,a,b)
		&=
		b_{01}+AP_1+BP_0,\\
		\mathcal H_{1,1}(n;m,a,b)
		&=
		b_{11}+AP_2+BP_1,\\
		\mathcal H_{2,1}(n;m,a,b)
		&=
		b_{21}+AP_3+BP_2,\\
		\mathcal H_{3,1}(n;m,a,b)
		&=
		b_{31}+AP_4+BP_3,\\[0.5ex]
		\mathcal H_{0,2}(n;m,a,b)
		&=
		b_{02}+2Ab_{11}+2Bb_{01}+A^2P_2+2ABP_1+B^2P_0,\\
		\mathcal H_{1,2}(n;m,a,b)
		&=
		b_{12}+2Ab_{21}+2Bb_{11}+A^2P_3+2ABP_2+B^2P_1,\\
		\mathcal H_{2,2}(n;m,a,b)
		&=
		b_{22}+2Ab_{31}+2Bb_{21}+A^2P_4+2ABP_3+B^2P_2,\\[0.5ex]
		\mathcal H_{0,3}(n;m,a,b)
		&=
		b_{03}
		+3Ab_{12}+3Bb_{02}
		+3A^2b_{21}+6ABb_{11}+3B^2b_{01}\\
		&\qquad
		+A^3P_3+3A^2BP_2+3AB^2P_1+B^3P_0,\\
		\mathcal H_{1,3}(n;m,a,b)
		&=
		b_{13}
		+3Ab_{22}+3Bb_{12}
		+3A^2b_{31}+6ABb_{21}+3B^2b_{11}\\
		&\qquad
		+A^3P_4+3A^2BP_3+3AB^2P_2+B^3P_1,\\[0.5ex]
		\mathcal H_{0,4}(n;m,a,b)
		&=
		b_{04}
		+4Ab_{13}+4Bb_{03}
		+6A^2b_{22}+12ABb_{12}+6B^2b_{02}\\
		&\qquad
		+4A^3b_{31}+12A^2Bb_{21}+12AB^2b_{11}+4B^3b_{01}\\
		&\qquad
		+A^4P_4+4A^3BP_3+6A^2B^2P_2+4AB^3P_1+B^4P_0.
	\end{align*}
	
	\subsection{Reciprocal step}
	
	Now assume
	\[
	0\le a,b<m.
	\]
	Set
	\[
	Y:=\Floor{\frac{a(n-1)+b}{m}}.
	\]
	If $Y=0$, then all ten moments vanish.
	
	Otherwise define
	\[
	g(t):=\Floor{\frac{mt+(m-b-1)}{a}},
	\qquad 0\le t\le Y-1.
	\]
	As usual, the reciprocal transformation exchanges the graph of
	\[
	f(x):=\Floor{\frac{ax+b}{m}}
	\]
	with the complementary staircase determined by $g$, and therefore expresses the moments for $(n;m,a,b)$ through the moments for $(Y;a,m,m-b-1)$.
	
	Write
	\[
	\mathbf G
	=
	\mathbf H(Y;a,m,m-b-1)
	=
	(g_{01},g_{11},g_{21},g_{31},g_{02},g_{12},g_{22},g_{03},g_{13},g_{04}),
	\]
	and let
	\[
	P_r:=P_r(n),
	\qquad
	Q_r:=P_r(Y).
	\]
	Then
	\begin{align*}
		\mathcal H_{0,1}(n;m,a,b)
		&=
		P_0Y-Q_0-g_{01},\\
		\mathcal H_{1,1}(n;m,a,b)
		&=
		\frac{2P_1Y-g_{01}-g_{02}}{2},\\
		\mathcal H_{2,1}(n;m,a,b)
		&=
		\frac{6P_2Y-g_{01}-3g_{02}-2g_{03}}{6},\\
		\mathcal H_{3,1}(n;m,a,b)
		&=
		\frac{4P_3Y-g_{02}-2g_{03}-g_{04}}{4},\\[0.5ex]
		\mathcal H_{0,2}(n;m,a,b)
		&=
		P_0Y^2-Q_0-g_{01}-2Q_1-2g_{11},\\
		\mathcal H_{1,2}(n;m,a,b)
		&=
		\frac{2P_1Y^2-g_{01}-g_{02}-2g_{11}-2g_{12}}{2},\\
		\mathcal H_{2,2}(n;m,a,b)
		&=
		\frac{6P_2Y^2-g_{01}-3g_{02}-2g_{03}-2g_{11}-6g_{12}-4g_{13}}{6},\\[0.5ex]
		\mathcal H_{0,3}(n;m,a,b)
		&=
		P_0Y^3-Q_0-g_{01}-3Q_1-3g_{11}-3Q_2-3g_{21},\\
		\mathcal H_{1,3}(n;m,a,b)
		&=
		\frac{2P_1Y^3-g_{01}-g_{02}-3g_{11}-3g_{12}-3g_{21}-3g_{22}}{2},\\
		\mathcal H_{0,4}(n;m,a,b)
		&=
		P_0Y^4-Q_0-g_{01}-4Q_1-4g_{11}-6Q_2-6g_{21}-4Q_3-4g_{31}.
	\end{align*}
	
	\begin{lemma}\label{lem:ten-moment-kernel}
		The ten moments
		\[
		\mathcal H_{0,1},\ \mathcal H_{1,1},\ \mathcal H_{2,1},\ \mathcal H_{3,1},\
		\mathcal H_{0,2},\ \mathcal H_{1,2},\ \mathcal H_{2,2},\
		\mathcal H_{0,3},\ \mathcal H_{1,3},\
		\mathcal H_{0,4}
		\]
		admit a recursive evaluation in $O(\log m)$ arithmetic operations.
	\end{lemma}
	
	\begin{proof}
		The affine step reduces $(a,b)$ modulo $m$.  The reciprocal step replaces $(n;m,a,b)$ by
		\[
		(Y;a,m,m-b-1),
		\qquad
		Y=\Floor{\frac{a(n-1)+b}{m}},
		\]
		so the first modulus strictly decreases from $m$ to $a<m$.  Hence the recursion has Euclidean depth $O(\log m)$.  Each step performs only $O(1)$ arithmetic operations on the ten stored moments and the power sums.
	\end{proof}

	\subsection{Complexity bound}
	
	Each recursive call either reduces $a$ and $b$ modulo $m$ (the affine step) or swaps the roles of $(a,m)$ with strictly smaller Euclidean parameters (the reciprocal step).  Hence the recursion depth is $O(\log m)=O(\log n)$.  Every step performs only $O(1)$ arithmetic operations on a constant-size family of moments.  This proves Corollary~\ref{cor:weighted-kernel-log}.

	\section{Implementation notes}\label{app:engineering}
	
	This appendix records two implementation-level simplifications for the ten-moment weighted-floor kernel. They do not change the asymptotic complexity, but together they reduce the constant factor substantially; in our implementation, the combined speedup is a little over $3\times$.
	
	First, we use explicit affine and reciprocal transition formulas instead of constructing the symbolic expansions at run time. This turns each update into a fixed straight-line computation and eliminates most temporary algebra.
	
	Second, we special-case the two most common affine subcases, namely $B=0$ and, within it, $A=1$. Since these patterns occur frequently in the Euclidean recursion, handling them with short handwritten formulas removes another layer of arithmetic overhead.

	\section{Proof of the two-term asymptotic expansion}\label{app:second-term}
	
	In this appendix we prove Theorem~\ref{thm:main-asymptotic}. We write
	\[
	S_1(n):=
	2\sum_{\substack{u>v\ge 1\\ (u,v)=1\\ u,v\le \sqrt n}}
	\ \sum_{\substack{a,b\ge 1\\ au+bv\le n\\ av+bu\le n}}
	(n-au-bv)(n-av-bu),
	\]
	\[
	S_2(n):=
	2\sum_{1\le a,b\le \sqrt n}
	\ \sum_{\substack{u>v\ge 1\\ (u,v)=1\\ au+bv\le n\\ av+bu\le n}}
	(n-au-bv)(n-av-bu),
	\]
	\[
	S_{12}(n):=
	2\sum_{\substack{u>v\ge 1\\ (u,v)=1\\ u,v\le \sqrt n}}
	\ \sum_{\substack{1\le a,b\le \sqrt n\\ au+bv\le n\\ av+bu\le n}}
	(n-au-bv)(n-av-bu).
	\]
	Thus $F(n)=S_1(n)+S_2(n)-S_{12}(n)+F_0(n)$. We derive the required expansions in the order $S_1$, $S_2$, $S_{12}$, and then assemble them. All remainders in this appendix are tracked only up to $o(n^4)$.
	For the linear constraints defining the summation regions, we use non-strict inequalities throughout, since on the boundary the weight $(n-au-bv)(n-av-bu)$ vanishes identically. Thus replacing a condition of the form $au+bv<n$ or $av+bu<n$ by the corresponding non-strict version does not change the summand at any lattice point. The only genuinely strict inequality that remains is the ordering constraint $u>v$; for the primitive-direction parameters below we work on the closed triangle $T=\{(x,y):0\le y\le x\le 1\}$.

	\subsection{\texorpdfstring{Shared inputs for $S_1$ and $S_2$}{Shared inputs for S1 and S2}}
	
	The next two lemmas are used in both the $S_1$ and $S_2$ analyses.
	
	\begin{lemma}[evaluation of the kernel]\label{lem:asym-kernel}
		For $p\ge q>0$, let
		\[
		K(p,q):=\iint_{\substack{\alpha,\beta>0\\ p\alpha+q\beta<1\\ q\alpha+p\beta<1}}
		(1-p\alpha-q\beta)(1-q\alpha-p\beta)\,d\alpha\,d\beta.
		\]
		Then $K(p,q)=\frac{3p-q}{12p^2(p+q)}=\frac{1}{3p(p+q)}-\frac{1}{12p^2}$ for $p>q$, while $K(p,p)=\frac{1}{12p^2}$.
	\end{lemma}
	
	\begin{proof}
		For $p>q$, set $x=p\alpha+q\beta$ and $y=q\alpha+p\beta$. The Jacobian is $p^2-q^2$, and the inverse map is
		\[
		\alpha=\frac{px-qy}{p^2-q^2},\qquad \beta=\frac{py-qx}{p^2-q^2}.
		\]
		Hence
		\[
		K(p,q)=\frac1{p^2-q^2}
		\iint_{\substack{0<x<1,\ 0<y<1\\ qx<py,\ qy<px}}(1-x)(1-y)\,dx\,dy.
		\]
		For fixed $x\in(0,1)$, the variable $y$ ranges over $\frac{q}{p}x<y<\min(\frac{p}{q}x,1)$, so splitting at $x=q/p$ gives
		\[
		K(p,q)=\frac1{p^2-q^2}
		\left(
		\int_0^{q/p}\!\int_{(q/p)x}^{(p/q)x}(1-x)(1-y)\,dy\,dx
		+\int_{q/p}^{1}\!\int_{(q/p)x}^{1}(1-x)(1-y)\,dy\,dx
		\right).
		\]
		The required elementary integrations are polynomial integrations over triangular regions, a special case of integration over simplices~\cite{BaldoniBerlineDeLoeraKoppeVergne}. Evaluating them yields $K(p,q)=\frac{3p-q}{12p^2(p+q)}$. For $p=q$, the region is $\alpha+\beta<1/p$, hence
		\[
		K(p,p)=\int_0^{1/p}\int_0^{1/p-\alpha}(1-p\alpha-p\beta)^2\,d\beta\,d\alpha=\frac1{12p^2}.\qedhere
		\]
	\end{proof}
	
	\begin{lemma}\label{lem:asym-harmonic}
		One has $\sum_{m\ge 1}(H_{2m-1}-H_m-\log 2)/m=\log^2 2-\pi^2/6$.
	\end{lemma}
	
	\begin{proof}
		Set $U:=\sum_{m\ge 1}(H_{2m}-H_m-\log 2)/m$. Since $H_{2m-1}=H_{2m}-1/(2m)$,
		\[
		\sum_{m\ge 1}\frac{H_{2m-1}-H_m-\log 2}{m}=U-\frac12\sum_{m\ge 1}\frac1{m^2}=U-\frac{\pi^2}{12}.
		\]
		Also $H_{2m}-H_m-\log 2=O(1/m)$, so the defining series for $U$ is absolutely convergent. Using the classical generating function \cite[Sec.~6.3]{ConcreteMath}
		\[
		\sum_{n\ge 1}\frac{H_n}{n}x^n=\operatorname{Li}_2(x)+\frac12\log^2(1-x)\qquad (|x|<1),
		\]
		we get for $0<x<1$,
		\[
		\sum_{m\ge 1}\frac{H_{2m}-H_m-\log 2}{m}x^m
		=F(\sqrt x)+F(-\sqrt x)-F(x)+\log 2\,\log(1-x),
		\]
		where $F(t):=\operatorname{Li}_2(t)+\frac12\log^2(1-t)$. The duplication identity
		\[
		\operatorname{Li}_2(t)+\operatorname{Li}_2(-t)=\frac12\operatorname{Li}_2(t^2)
		\]
		(compare \cite[Eq.~25.12.12]{DLMF}) shows that the right-hand side tends to $\log^2 2-\pi^2/12$ as $x\to 1^-$. Since the coefficients are absolutely summable, the series on the left converges absolutely at $x=1$ and uniformly for $x\in[0,1]$. Therefore one may pass to the limit $x\to1^-$ termwise, and the left-hand side tends to $U$. Hence $U=\log^2 2-\pi^2/12$, and the claim follows.
	\end{proof}
	
	\begin{remark}[parametric Euler--Maclaurin patching]\label{rem:parametric-EM}
		We use a standard local form of weighted Euler--Maclaurin summation on polytopal families; compare, for example, the polytope expansions in~\cite{BrandoliniColzaniGariboldiGiganteMonguzzi}. Let $\Theta$ be a compact parameter set, and suppose that it is covered by finitely many relatively open cells $\Theta_\nu$ such that on each cell the family of polygons $P_\theta$ has fixed combinatorial type, the defining affine inequalities depend $C^2$-smoothly on $\theta\in\Theta_\nu$, and the weights $w_\theta\in C^2(P_\theta)$ depend $C^2$-smoothly on $\theta$ with bounds uniform on compact subsets of $\Theta_\nu$. On each cell, the weighted Euler--Maclaurin formula therefore gives an expansion
		\[
		\sum_{m\in TP_\theta\cap\mathbb Z^2} w_\theta(m/T)=T^2A_\nu(\theta)+TB_\nu(\theta)+O(1)
		\qquad (T\ge 1),
		\]
		uniformly for $\theta$ in compact subsets of $\Theta_\nu$.
		
		In the applications below, one checks directly that the area coefficient
		\[
		A(\theta):=\iint_{P_\theta} w_\theta(x)\,dx
		\]
		extends continuously and remains bounded on all of $\Theta$, and, when the linear term is needed, that the boundary coefficient
		\[
		B(\theta):=\frac12\sum_{E\subset\partial P_\theta}\int_E w_\theta\,d\sigma_E
		\]
		does as well. Once these bounded continuous extensions across the finitely many transition strata have been verified, the cellwise expansions may be patched into a single uniform estimate on $\Theta$ of the same shape:
		\[
		\sum_{m\in TP_\theta\cap\mathbb Z^2} w_\theta(m/T)=T^2A(\theta)+TB(\theta)+O(1),
		\]
		or, if only the area term is used, simply
		\[
		\sum_{m\in TP_\theta\cap\mathbb Z^2} w_\theta(m/T)=T^2A(\theta)+O(T).
		\]
		This is exactly the way Euler--Maclaurin is invoked in the analyses of $S_1$, $S_2$, and $S_{12}$.
	\end{remark}
	
	\begin{remark}[replacing $\sqrt n$ by $\lfloor\sqrt n\rfloor$]\label{rem:floor-sqrt-safe}
		Let $N:=\lfloor\sqrt n\rfloor$. Then
		\[
		N=\sqrt n+O(1),
		\qquad
		\log N=\frac12\log n+o(1),
		\qquad
		N^8=n^4+o(n^4).
		\]
		In particular,
		\[
		n^3\log n=o(n^4),\qquad n^3N=o(n^4),\qquad N^7=o(n^4),\qquad N^7\log N=o(n^4).
		\]
	\end{remark}
	
	\subsection{\texorpdfstring{The contribution $S_1$}{The contribution S1}}
	
	\begin{lemma}\label{lem:asym-phi}
		As $N\to\infty$, one has $\sum_{u\le N}\varphi(u)/u^2=\frac{6}{\pi^2}\log N+\bigl(\frac{6\gamma}{\pi^2}-\frac{\zeta'(2)}{\zeta(2)^2}\bigr)+o(1)$.
	\end{lemma}
	
	\begin{proof}
		Put $a_n:=\varphi(n)/n$. Then
		\[
		\sum_{n\ge 1}\frac{a_n}{n^s}=\sum_{n\ge 1}\frac{\varphi(n)}{n^{s+1}}=\frac{\zeta(s)}{\zeta(s+1)} \qquad (\Re s>1).
		\]
		As $s\to 1^+$,
		\[
		\frac{\zeta(s)}{\zeta(s+1)}=\frac{6}{\pi^2}\frac1{s-1}+\left(\frac{6\gamma}{\pi^2}-\frac{\zeta'(2)}{\zeta(2)^2}\right)+O(s-1).
		\]
		By Delange's theorem for Dirichlet series with a simple pole (see \cite[Ch.~II.5, Th.~3]{Tenenbaum}; compare also \cite[Ch.~5, \S1]{MontgomeryVaughan}), this implies
		\[
		\sum_{n\le N}\frac{a_n}{n}=\sum_{n\le N}\frac{\varphi(n)}{n^2}
		=\frac{6}{\pi^2}\log N+\left(\frac{6\gamma}{\pi^2}-\frac{\zeta'(2)}{\zeta(2)^2}\right)+o(1),
		\]
		as claimed.
	\end{proof}
	
	\begin{lemma}[kernel row sum]\label{lem:asym-kernel-rows}
		For $u\ge 2$, let $k(u):=\sum_{\substack{1\le v<u\\ (u,v)=1}}K(u,v)$. Then
		\[
		k(u)=\frac{4\log 2-1}{12}\frac{\varphi(u)}{u^2}+r(u),\qquad r(u)=O\!\left(\frac{\tau(u)}{u^2}\right),
		\]
		and therefore
		\[
		\sum_{u\le N}k(u)=\frac{4\log 2-1}{2\pi^2}\log N+D_1+o(1),
		\]
		where
		\[
		D_1=\frac{4\log 2-1}{12}\left(\frac{6\gamma}{\pi^2}-\frac{\zeta'(2)}{\zeta(2)^2}-1\right)+\frac{2\log^2 2}{\pi^2}+\frac{\log 2}{3}-\frac13.
		\]
	\end{lemma}
	
	\begin{proof}
		By Lemma~\ref{lem:asym-kernel},
		\[
		k(u)=\frac1{3u}\sum_{\substack{1\le v<u\\ (u,v)=1}}\frac1{u+v}-\frac{\varphi(u)}{12u^2}.
		\]
		Using inclusion--exclusion,
		\[
		\sum_{\substack{1\le v<u\\ (u,v)=1}}\frac1{u+v}
		=\sum_{d\mid u}\mu(d)\sum_{\substack{1\le v<u\\ d\mid v}}\frac1{u+v}
		=\sum_{d\mid u}\frac{\mu(d)}{d}\sum_{1\le m<u/d}\frac1{u/d+m}.
		\]
		Since $\sum_{1\le m<M}(M+m)^{-1}=H_{2M-1}-H_M$, this becomes
		\[
		\sum_{\substack{1\le v<u\\ (u,v)=1}}\frac1{u+v}
		=\sum_{d\mid u}\frac{\mu(d)}{d}\Bigl(H_{2u/d-1}-H_{u/d}\Bigr).
		\]
		Hence
		\[
		k(u)=\frac{4\log 2-1}{12}\frac{\varphi(u)}{u^2}
		+\frac1{3u}\sum_{d\mid u}\frac{\mu(d)}{d}\Bigl(H_{2u/d-1}-H_{u/d}-\log 2\Bigr).
		\]
		The displayed remainder is $O(\tau(u)/u^2)$ because $H_{2m-1}-H_m-\log 2=O(1/m)$ uniformly in $m\ge 1$.
		Let
		\[
		R(N):=\sum_{u\le N}\frac1{3u}\sum_{d\mid u}\frac{\mu(d)}{d}\Bigl(H_{2u/d-1}-H_{u/d}-\log 2\Bigr).
		\]
		Setting $u=dm$ and exchanging the order of summation gives
		\[
		R(N)=\frac13\sum_{m\le N}\frac{H_{2m-1}-H_m-\log 2}{m}\sum_{d\le N/m}\frac{\mu(d)}{d^2}.
		\]
		Since $H_{2m-1}-H_m-\log 2=O(1/m)$, the series
		\[
		\sum_{m\ge 1}\frac{|H_{2m-1}-H_m-\log 2|}{m}
		\]
		converges. Therefore the rearrangement above is absolutely summable, and the passage to the limit below is justified by dominated convergence. Since $\sum_{d\le X}\mu(d)/d^2=1/\zeta(2)+O(X^{-1})$, we obtain
		\[
		R(N)=\frac{1}{3\zeta(2)}\sum_{m\ge 1}\frac{H_{2m-1}-H_m-\log 2}{m}+o(1)
		=\frac{2\log^2 2}{\pi^2}-\frac16+o(1).
		\]
		Combining this with Lemma~\ref{lem:asym-phi} gives
		\[
		\sum_{u\le N}k(u)=\frac{4\log 2-1}{2\pi^2}\log N+\frac{4\log 2-1}{12}\Bigl(\frac{6\gamma}{\pi^2}-\frac{\zeta'(2)}{\zeta(2)^2}\Bigr)+\frac{2\log^2 2}{\pi^2}-\frac16+o(1).
		\]
		Since the lemma sums over $u\ge 2$, we must subtract the $u=1$ contribution $(4\log 2-1)/12$. This yields exactly the stated constant $D_1$.
	\end{proof}
	
	\begin{theorem}[asymptotic for $S_1$]\label{thm:asym-S1}
		One has
		\[
		S_1(n)=A_1n^4\log n+B_1n^4+o(n^4),
		\]
		where $A_1=\frac{4\log 2-1}{2\pi^2}$ and
		\[
		B_1=\frac{4\log 2-1}{6}\left(\frac{6\gamma}{\pi^2}-\frac{\zeta'(2)}{\zeta(2)^2}-1\right)+\frac{4\log^2 2}{\pi^2}+\frac{2\log 2}{3}-\frac23.
		\]
	\end{theorem}

	\begin{proof}
		Let $N:=\lfloor\sqrt n\rfloor$.
		For fixed primitive $u>v$, define
		\[
		G_n(u,v):=\sum_{\substack{a,b\ge 1\\ au+bv\le n\\ av+bu\le n}}(n-au-bv)(n-av-bu).
		\]
		Write
		\[
		P_{u,v}:=\{(\alpha,\beta)\in\mathbb R_{\ge 0}^2: u\alpha+v\beta\le 1,\ v\alpha+u\beta\le 1\}
		\]
		and
		\[
		w_{u,v}(\alpha,\beta):=(1-u\alpha-v\beta)(1-v\alpha-u\beta).
		\]
		Then
		\[
		G_n(u,v)=n^2\sum_{(a,b)\in nP_{u,v}\cap \mathbb Z_{\ge 1}^2} w_{u,v}(a/n,b/n).
		\]
		After the normalization $(\tilde\alpha,\tilde\beta)=(u\alpha,u\beta)$ and $\lambda:=v/u\in[0,1]$, the polygons become
		\[
		\widetilde P_\lambda:=\{(\tilde\alpha,\tilde\beta)\in\mathbb R_{\ge 0}^2: \tilde\alpha+\lambda\tilde\beta\le 1,\ \lambda\tilde\alpha+\tilde\beta\le 1\},
		\]
		and the weights become
		\[
		\widetilde w_\lambda(\tilde\alpha,\tilde\beta):=(1-\tilde\alpha-\lambda\tilde\beta)(1-\lambda\tilde\alpha-\tilde\beta).
		\]
		Thus $(\widetilde P_\lambda,\widetilde w_\lambda)$ form a one-parameter piecewise-$C^2$ family for $\lambda\in[0,1]$. As $\lambda\to1$, the polygon converges continuously to the limiting triangle $\{\tilde\alpha,\tilde\beta\ge0,\ \tilde\alpha+\tilde\beta\le1\}$, and the weight remains uniformly $C^2$. In particular, the area term and the linear boundary term stay uniformly bounded and vary continuously up to $\lambda=1$. Therefore Remark~\ref{rem:parametric-EM} yields, after checking the continuous extension of the area and boundary coefficients to $\lambda=1$, uniformly for $u>v\ge1$,
		\[
		n^2\sum_{(a,b)\in nP_{u,v}\cap \mathbb Z_{\ge 0}^2} w_{u,v}(a/n,b/n)
		=n^4K(u,v)+n^3L(u,v)+O(n^2),
		\]
		where the two oblique sides contribute $0$ because $w_{u,v}$ vanishes identically there, while the horizontal and vertical sides contribute
		\[
		\int_0^{1/u}(1-ut)(1-vt)\,dt=\frac{3u-v}{6u^2}
		\]
		each. Hence
		\[
		L(u,v)=\frac12\cdot 2\int_0^{1/u}(1-ut)(1-vt)\,dt=\frac{3u-v}{6u^2}.
		\]
		Now pass from $\mathbb Z_{\ge 0}^2$ to $\mathbb Z_{\ge 1}^2$ by deleting the two coordinate-axis slices and then adding back the doubly deleted corner $(0,0)$. Along the horizontal axis one has
		\[
		n^2\sum_{0\le b\le n/u} w_{u,v}(0,b/n)
		= n^2\sum_{0\le b\le n/u}(1-vb/n)(1-ub/n)
		= n^3\int_0^{1/u}(1-vt)(1-ut)\,dt+O(n^2),
		\]
		by the one-dimensional Euler--Maclaurin formula, and the same estimate holds on the vertical axis. The one-dimensional Euler--Maclaurin remainder here is uniform in $u,v\le N$: the interval has length $1/u$, while for $f_{u,v}(t):=(1-vt)(1-ut)$ one has $|f_{u,v}(t)|\le 1$ and $|f_{u,v}^{\prime\prime}(t)|=2uv$ on $[0,1/u]$. Thus the usual one-dimensional remainder is $O(1)$ before the outside factor $n^2$, hence $O(n^2)$ after rescaling. Hence the two deleted slices together subtract
		\[
		2n^3\int_0^{1/u}(1-vt)(1-ut)\,dt+O(n^2)=\frac{3u-v}{3u^2}n^3+O(n^2)
		\]
		from the expansion over $\mathbb Z_{\ge 0}^2$, while the corner $(0,0)$ contributes only $n^2w_{u,v}(0,0)=n^2$. Therefore
		\[
		G_n(u,v)=n^4K(u,v)-\frac{3u-v}{6u^2}n^3+O(n^2),
		\]
		uniformly for $u>v\ge 1$ and $u,v\le N$.
		Summing over primitive $u>v$ yields
		\[
		S_1(n)=2n^4\sum_{u\le N}k(u)-n^3\sum_{u\le N}\sum_{\substack{1\le v<u\\ (u,v)=1}}\frac{3u-v}{3u^2}+o(n^4).
		\]
		Since $\#\{(u,v):1\le v<u\le N,\ (u,v)=1\}=O(N^2)=O(n)$, the accumulated $O(n^2)$ error from the individual expansions is $O(n^3)=o(n^4)$, which is absorbed into the remainder above. Since $\sum_{\substack{1\le v<u\\ (u,v)=1}} v=u\varphi(u)/2$ for $u>1$,
		\[
		\sum_{\substack{1\le v<u\\ (u,v)=1}}\frac{3u-v}{3u^2}=\frac56\frac{\varphi(u)}{u},
		\]
		so
		\[
		S_1(n)=2n^4\sum_{u\le N}k(u)-\frac56n^3\sum_{u\le N}\frac{\varphi(u)}{u}+o(n^4).
		\]
		The second sum is $o(n^4)$ by Remark~\ref{rem:floor-sqrt-safe}. Therefore Lemma~\ref{lem:asym-kernel-rows} gives
		\[
		2n^4\sum_{u\le N}k(u)=A_1n^4\log n+B_1n^4+o(n^4),
		\]
		and the claimed asymptotic follows.
	\end{proof}
	
	\subsection{\texorpdfstring{The contribution $S_2$}{The contribution S2}}
	
	Throughout this subsection we write
	\[
	M(x):=\sum_{d\le x}\mu(d),\qquad
	\varepsilon(x):=\exp\!\Bigl(-c(\log x)^{3/5}(\log\log x)^{-1/5}\Bigr),
	\]
	where $c>0$ is an absolute constant. By the classical Walfisz estimate for the Mertens function \cite[Ch.~I.4]{Tenenbaum},
	\[
	M(x)=O\!\bigl(x\varepsilon(x)\bigr).
	\]
	
	\begin{lemma}[summed second-order input for $S_2$]\label{lem:key-S2-input}
		Let $N:=\lfloor \sqrt n \rfloor$. For $1\le a,b\le N$, define
		\[
		\begin{aligned}
			H_n(a,b)&:=\sum_{\substack{u>v\ge 1\\ (u,v)=1\\ au+bv\le n\\ av+bu\le n}}(n-au-bv)(n-av-bu),\\
			\mathcal J(a,b)&:=\iint_{\substack{x>y>0\\ ax+by\le 1\\ ay+bx\le 1}}(1-ax-by)(1-ay-bx)\,dx\,dy.
		\end{aligned}
		\]
		Then
		\[
		\sum_{1\le a,b\le N} H_n(a,b)=\frac{6}{\pi^2}n^4\sum_{1\le a,b\le N}\mathcal J(a,b)+o(n^4).
		\]
	\end{lemma}
	
	\begin{proof}
		For $t\ge a+b$, let
		\[
		F_t(a,b):=\sum_{\substack{u>v\ge 1\\ au+bv\le t\\ av+bu\le t}}(t-au-bv)(t-av-bu),
		\]
		where the sum is over all $u>v\ge 1$, without the coprimality condition. For each pair $(a,b)$, define
		\[
		X_{a,b}:=\Bigl\lfloor\frac{n}{a+b}\Bigr\rfloor.
		\]
		Then M\"obius inversion gives
		\[
		H_n(a,b)=\sum_{d\le X_{a,b}}\mu(d)d^2F_{n/d}(a,b).
		\]
		Put $m:=\max(a,b)$, $p:=a/m$, $q:=b/m$, and $\tau:=t/m$. Then $(p,q)$ ranges over the compact set
		\[
		\Theta:=\{(p,q)\in[0,1]^2:\max(p,q)=1\},
		\]
		and
		\[
		\begin{aligned}
			\Omega_{p,q}&:=\{(x,y):x>y>0,\ px+qy\le 1,\ qx+py\le 1\},\\
			w_{p,q}(x,y)&:=(1-px-qy)(1-qx-py).
		\end{aligned}
		\]
		give the rescaled representation
		\[
		F_t(a,b)=t^2\sum_{(u,v)\in \tau\Omega_{p,q}\cap\mathbb Z^2} w_{p,q}(u/\tau,v/\tau).
		\]
		The family $(\Omega_{p,q},w_{p,q})_{(p,q)\in\Theta}$ is compact and piecewise-$C^2$. Let
		\[
		\overline\Omega_{p,q}:=\{(x,y):x\ge y\ge 0,\ px+qy\le 1,\ qx+py\le 1\}.
		\]
		Passing from the open wedge $x>y>0$ to the closed wedge $x\ge y\ge 0$ adds only points on the diagonal and coordinate axes. For each fixed $(p,q)$, these boundary slices contain $O(\tau)$ lattice points and the weight is uniformly bounded, so the resulting correction to $F_t(a,b)$ is $t^2\cdot O(\tau)=O(t^3/m)=O(t^3)$, which is harmless for the final summation. Thus we may invoke Euler--Maclaurin on $\overline\Omega_{p,q}$.
		Away from the endpoint $(p,q)=(1,1)$ the combinatorial type is constant on each branch of $\Theta$, so Remark~\ref{rem:parametric-EM} applies directly there. At $(p,q)=(1,1)$ the quadrilateral degenerates continuously to the limiting triangle. This causes no problem for the coefficient bounds used below: the two oblique sides satisfy $w_{p,q}=0$ identically, while the remaining boundary pieces lie on $x=y$ and on the coordinate axes; along these edges both the edge lengths and the restricted weights stay uniformly bounded as $(p,q)\to(1,1)$. Hence the area coefficient extends continuously and remains uniformly bounded on all of $\Theta$; since only this coefficient is used below, Remark~\ref{rem:parametric-EM} gives, uniformly in $(p,q)\in\Theta$ and $\tau\ge 1$,
		\[
		\sum_{(u,v)\in \tau\overline\Omega_{p,q}\cap\mathbb Z^2} w_{p,q}(u/\tau,v/\tau)=\tau^2A(p,q)+O(\tau),
		\]
		with implied constants independent of $(p,q)$ and $\tau$. Therefore
		\[
		F_t(a,b)=t^4\mathcal J(a,b)+O\!\left(\frac{t^3}{\max(a,b)}\right)+O(t^2),
		\]
		where $\mathcal J(a,b)=A(p,q)/m^2$, uniformly for $1\le a,b\le N$ and $t\ge a+b$.
		Substituting this expansion into the M\"obius formula and using
		\[
		\sum_{d\le x}\frac{\mu(d)}{d}=O\!\bigl(\varepsilon(x)\bigr),\qquad
		\sum_{d\le x}\frac{\mu(d)}{d^2}=\frac1{\zeta(2)}+O\!\left(\frac{\varepsilon(x)}{x}\right),\qquad
		M(x)=\sum_{d\le x}\mu(d)=O\!\bigl(x\varepsilon(x)\bigr),
		\]
		we obtain
		\[
		H_n(a,b)=\frac{6}{\pi^2}n^4\mathcal J(a,b)
		+O\!\left(\varepsilon_n\frac{n^4\mathcal J(a,b)}{X_{a,b}}\right)
		+O\!\left(\varepsilon_n n^3\frac{1}{\max(a,b)}\right)
		+O\!\left(n^2\bigl|M(X_{a,b})\bigr|\right),
		\]
		where $\varepsilon_n:=\varepsilon(\sqrt n/2)$ and we used $X_{a,b}\ge \frac12\sqrt n$ for $a,b\le N=\lfloor\sqrt n\rfloor$.
		Since $X_{a,b}=\lfloor n/(a+b)\rfloor$, one has for all sufficiently large $n$,
		\[
		\frac{n}{2(a+b)}\le X_{a,b}\le \frac{n}{a+b}.
		\]
		Moreover $\mathcal J(a,b)=\frac12K(\max(a,b),\min(a,b))$, so Lemma~\ref{lem:asym-kernel} implies
		\[
		\sum_{a,b\le N}(a+b)\mathcal J(a,b)=O(N).
		\]
		Also,
		\[
		\sum_{a,b\le N}\frac{1}{\max(a,b)}\le C_2\sum_{m\le N}\frac{2m-1}{m}=O(N).
		\]
		Using again $M(X_{a,b})=O(X_{a,b}\varepsilon_n)$ uniformly in $a,b$, we obtain
		\[
		\sum_{a,b\le N}\varepsilon_n\frac{n^4\mathcal J(a,b)}{X_{a,b}}=O(\varepsilon_n n^3N),
		\qquad
		\sum_{a,b\le N}\varepsilon_n n^3\frac{1}{\max(a,b)}=O(\varepsilon_n n^3N),
		\]
		and
		\[
		\sum_{a,b\le N} n^2\bigl|M(X_{a,b})\bigr|
		\le C_4\varepsilon_n n^2 \sum_{a,b\le N} X_{a,b}
		\le C_5\varepsilon_n n^3\sum_{a,b\le N}\frac{1}{a+b}
		=O(\varepsilon_n n^3 N)=o(n^4).
		\]
		Since $N=\lfloor\sqrt n\rfloor$, we have $n^3N=n^{7/2}=o(n^4)$, and since $\varepsilon_n\to 0$, the first two displayed sums are also $o(n^4)$. This proves the claim.
	\end{proof}
	
	\begin{theorem}[asymptotic for $S_2$]\label{thm:asym-S2}
		One has
		\[
		S_2(n)=A_2n^4\log n+B_2n^4+o(n^4),
		\]
		where $A_2=\frac{4\log 2-1}{2\pi^2}$ and $B_2=\frac{(4\log 2-1)\gamma+4\log^2 2}{\pi^2}-\frac{5}{12}$.
	\end{theorem}
	
	\begin{proof}
		By definition,
		\[
		S_2(n)=2\sum_{1\le a,b\le N}H_n(a,b),\qquad N:=\lfloor \sqrt n\rfloor.
		\]
		By Lemma~\ref{lem:key-S2-input},
		\[
		S_2(n)=\frac{12}{\pi^2}n^4\sum_{1\le a,b\le N}\mathcal J(a,b)+o(n^4).
		\]
		Since $\mathcal J(a,b)=\frac12K(\max(a,b),\min(a,b))$, we have
		\[
		\sum_{1\le a,b\le N}\mathcal J(a,b)=\sum_{u\le N}\sum_{v<u}K(u,v)+\frac12\sum_{u\le N}K(u,u).
		\]
		A direct summation using Lemma~\ref{lem:asym-kernel} gives
		\[
		\sum_{v=1}^{u-1}K(u,v)=\frac{H_{2u-1}-H_u}{3u}-\frac{u-1}{12u^2},\qquad K(u,u)=\frac{1}{12u^2}.
		\]
		Summing over $u\le N$, using Lemma~\ref{lem:asym-harmonic}, and also
		\[
		\sum_{u\le N}\frac{1}{u^2}=\frac{\pi^2}{6}-\frac{1}{N}+O(N^{-2}),
		\]
		one obtains
		\[
		\sum_{u\le N}\sum_{v<u}K(u,v)+\frac12\sum_{u\le N}K(u,u)=\frac{4\log 2-1}{12}\log N+C_0+o(1),
		\]
		where $C_0=\frac{4\log 2-1}{12}\gamma+\frac{\log^2 2}{3}-\frac{5\pi^2}{144}$.
		Multiplying by $(12/\pi^2)n^4$ yields
		\[
		S_2(n)=\frac{4\log 2-1}{\pi^2}n^4\log N+\frac{12C_0}{\pi^2}n^4+o(n^4).
		\]
		Since $N=\lfloor\sqrt n\rfloor$, one has $N/\sqrt n\to 1$, hence
		\[
		\log N=\frac12\log n+o(1).
		\]
		Therefore
		\[
		S_2(n)=A_2n^4\log n+B_2n^4+o(n^4),
		\]
		with $B_2=12C_0/\pi^2=\frac{(4\log 2-1)\gamma+4\log^2 2}{\pi^2}-\frac{5}{12}$.
	\end{proof}
	
	\subsection{\texorpdfstring{The contribution $S_{12}$}{The contribution S12}}
	
	\begin{lemma}[weighted primitive points]\label{lem:primitive-weighted}
		Let
		\[
		T:=\{(x,y)\in\mathbb R^2:0\le y\le x\le 1\}.
		\]
		Assume that $g$ is continuous on $T$, globally Lipschitz on $T$, and that there is a finite partition of $T$ into polygons with piecewise $C^2$ boundaries such that $g$ is $C^2$ on each cell and its first derivatives are uniformly bounded there. Then
		\[
		\sum_{\substack{u>v\ge 1\\ (u,v)=1\\ u,v\le N}}g\!\left(\frac{u}{N},\frac{v}{N}\right)
		=\frac{6}{\pi^2}N^2\iint_T g(x,y)\,dx\,dy+O_g(N\log N).
		\]
	\end{lemma}
	
	\begin{proof}
		By M\"obius inversion,
		\[
		\sum_{\substack{u>v\ge 1\\ (u,v)=1\\ u,v\le N}}g\!\left(\frac{u}{N},\frac{v}{N}\right)=\sum_{d\le N}\mu(d)\sum_{\substack{m>n\ge 1\\ m,n\le N/d}}g\!\left(\frac{dm}{N},\frac{dn}{N}\right).
		\]
		Let $M:=N/d$ and $M_d:=\lfloor M\rfloor$. Then the inner sum is
		\[
		\sum_{\substack{m>n\ge 1\\ m,n\le M_d}}g\!\left(\frac{m}{M},\frac{n}{M}\right).
		\]
		Replacing the cutoff $M_d$ by $M$ changes the summation region only in a boundary strip containing $O(M)$ lattice points, hence contributes $O_g(M)$. On the common part, the Lipschitz estimate and $|M^{-1}-M_d^{-1}|=O(M^{-2})$ give
		\[
		\left|g\!\left(\frac{m}{M},\frac{n}{M}\right)-g\!\left(\frac{m}{M_d},\frac{n}{M_d}\right)\right|=O_g(M^{-1})
		\]
		for each lattice point, hence another total $O_g(M)$. Therefore the inner sum differs by at most $O_g(M)$ from
		\[
		\sum_{\substack{m>n\ge 1\\ m,n\le M_d}}g\!\left(\frac{m}{M_d},\frac{n}{M_d}\right).
		\]
		Applying the weighted polygonal Euler--Maclaurin expansion separately on the finitely many cells of the partition of the closed triangle $T$ yields
		\[
		\sum_{\substack{m>n\ge 1\\ m,n\le M_d}}g\!\left(\frac{m}{M_d},\frac{n}{M_d}\right)=A_gM_d^2+O_g(M_d),
		\qquad A_g:=\iint_T g(x,y)\,dx\,dy,
		\]
		uniformly in $M_d\ge 1$. Since $M_d=M+O(1)$, this becomes
		\[
		\sum_{\substack{m>n\ge 1\\ m,n\le M_d}}g\!\left(\frac{m}{M},\frac{n}{M}\right)=A_gM^2+O_g(M).
		\]
		Hence
		\[
		\sum_{\substack{u>v\ge 1\\ (u,v)=1\\ u,v\le N}}g\!\left(\frac{u}{N},\frac{v}{N}\right)
		=N^2A_g\sum_{d\le N}\frac{\mu(d)}{d^2}+O_g\!\left(N\sum_{d\le N}\frac{|\mu(d)|}{d}\right).
		\]
		Using
		\[
		\sum_{d\le N}\frac{\mu(d)}{d^2}=\frac1{\zeta(2)}+O\!\left(\frac1N\right),
		\qquad
		\sum_{d\le N}\frac{|\mu(d)|}{d}=O(\log N),
		\]
		where the second bound is only the standard squarefree-indicator estimate, we obtain the claim. In particular, the $N\log N$ remainder comes exactly from the factor $N\sum_{d\le N}|\mu(d)|/d$.
	\end{proof}
	
	\begin{theorem}[asymptotic for $S_{12}$]\label{thm:asym-S12}
		One has
		\[
		S_{12}(n)=B_{12}n^4+o(n^4),
		\]
		where $B_{12}=-\frac13+\frac{2\log^2 2}{\pi^2}+\frac{19\log 2}{3\pi^2}-\frac{1}{12\pi^2}$.
	\end{theorem}
	
	\begin{proof}
		Let $N:=\lfloor\sqrt n\rfloor$ and put $x:=u/N$, $y:=v/N$. For primitive $u>v$ with $u,v\le N$, write
		\[
		Q_N(u,v):=\sum_{\substack{1\le a,b\le N\\ au+bv\le n\\ av+bu\le n}}(n-au-bv)(n-av-bu).
		\]
		After the scaling $a=N\alpha$ and $b=N\beta$, the admissible region becomes
		\[
		D_{x,y}:=[0,1]^2\cap\{x\alpha+y\beta\le 1\}\cap\{y\alpha+x\beta\le 1\}.
		\]
		Therefore, with $\rho_n:=n/N^2$, one has
		\[
		Q_N(u,v)=N^4\sum_{(\alpha,\beta)\in D_{x,y}^{(\rho_n)}\cap (N^{-1}\mathbb Z)^2}(\rho_n-x\alpha-y\beta)(\rho_n-y\alpha-x\beta),
		\]
		where
		\[
		D_{x,y}^{(\rho)}:=[0,1]^2\cap\{x\alpha+y\beta\le \rho\}\cap\{y\alpha+x\beta\le \rho\}.
		\]
		By Remark~\ref{rem:floor-sqrt-safe}, $\rho_n=1+o(1)$.
		Write
		\[
		T_0:=T\cap\{x+y\le1\},\qquad T_1:=T\cap\{x+y\ge1\}.
		\]
		On the relative interiors of $T_0$ and $T_1$, and for $\rho$ in a neighborhood of $1$, the family $(D_{x,y}^{(\rho)},(\rho-x\alpha-y\beta)(\rho-y\alpha-x\beta))$ has fixed combinatorial type and depends piecewise-$C^2$ on $(x,y,\rho)$ with uniform bounds. Hence Remark~\ref{rem:parametric-EM} applies separately on each relative interior cell and gives there the uniform area-term expansion
		\[
		Q_N(u,v)=N^6I_{\rho_n}(x,y)+O(N^5),
		\]
		where
		\[
		I_{\rho}(x,y):=\iint_{D_{x,y}^{(\rho)}}(\rho-x\alpha-y\beta)(\rho-y\alpha-x\beta)\,d\alpha\,d\beta.
		\]
		Along the interface $x+y=1$ the combinatorial type changes, but the domains vary continuously, and the only edges that appear or disappear are the oblique edges on which the weight vanishes identically. Hence the area coefficient matches continuously across the interface and remains uniformly bounded there, so the same expansion holds uniformly for all $(x,y)\in T$ and $\rho$ near $1$.
		Since the family is piecewise $C^1$ in $\rho$ near $1$, uniformly in $(x,y)$ one has
		\[
		I_{\rho_n}(x,y)=I(x,y)+O(|\rho_n-1|),\qquad I:=I_1.
		\]
		The resulting change in the main contribution to $S_{12}(n)$ is bounded by
		\[
		N^6\cdot O(|\rho_n-1|)\cdot \#\{(u,v):u>v\ge 1,\ u,v\le N,\ (u,v)=1\}=O(N^8|\rho_n-1|)=o(n^4),
		\]
		because $\#\{(u,v):u>v\ge 1,\ u,v\le N,\ (u,v)=1\}=O(N^2)$, $N^8=n^4+o(n^4)$, and $\rho_n-1=o(1)$. Thus it is enough to work with $\rho=1$, that is, with
		\[
		D_{x,y}=D_{x,y}^{(1)}=[0,1]^2\cap\{x\alpha+y\beta\le 1\}\cap\{y\alpha+x\beta\le 1\}.
		\]
		Thus $I$ is continuous on $T$. On the relative interiors of $T_0$ and $T_1$, $I$ is piecewise $C^1$ with bounded first derivatives, as follows from the explicit formulas below; hence $I$ is globally Lipschitz on $T$ and satisfies the assumptions of Lemma~\ref{lem:primitive-weighted}.
		Applying Lemma~\ref{lem:primitive-weighted} to $I$ gives
		\[
		N^6\sum_{\substack{u>v\ge 1\\ (u,v)=1\\ u,v\le N}} I\!\left(\frac{u}{N},\frac{v}{N}\right)=\frac{12}{\pi^2}N^8\iint_T I(x,y)\,dx\,dy+O(N^7\log N).
		\]
		The uniform $O(N^5)$ remainder contributes at most
		\[
		O(N^5)\cdot \#\{(u,v):u>v\ge 1,\ u,v\le N,\ (u,v)=1\}=O(N^7)=o(n^4).
		\]
		Therefore
		\[
		S_{12}(n)=\frac{12}{\pi^2}N^8\iint_T I(x,y)\,dx\,dy+O(N^7\log N).
		\]
		To evaluate the bulk integral, split $T$ into the two regions $T_0:=T\cap\{x+y\le 1\}$ and $T_1:=T\cap\{x+y>1\}$. On $T_0$ one has $D_{x,y}=[0,1]^2$, hence
		\[
		I(x,y)=\int_0^1\!\!\int_0^1 (1-x\alpha-y\beta)(1-y\alpha-x\beta)\,d\alpha\,d\beta
		=1-x-y+\frac{x^2+y^2}{4}+\frac{2xy}{3}.
		\]
		On $T_1$, with
		\[
		W_{x,y}(\alpha,\beta):=(1-x\alpha-y\beta)(1-y\alpha-x\beta),
		\]
		and
		\[
		\beta_1(\alpha):=\frac{1-x\alpha}{y},\qquad \beta_2(\alpha):=\frac{1-y\alpha}{x},
		\]
		one has
		\[
		\begin{aligned}
			I(x,y)={}&\int_0^{\frac{1-x}{y}}\!\int_0^1 W_{x,y}(\alpha,\beta)\,d\beta\,d\alpha\\
			&+\int_{\frac{1-x}{y}}^{\frac{1-y}{x}}\!\int_0^{\beta_1(\alpha)} W_{x,y}(\alpha,\beta)\,d\beta\,d\alpha\\
			&+\int_{\frac{1-y}{x}}^{1}\!\int_0^{\beta_2(\alpha)} W_{x,y}(\alpha,\beta)\,d\beta\,d\alpha,
		\end{aligned}
		\]
		and evaluating these elementary integrals gives
		\[
		\begin{aligned}
			I(x,y)={}&1-x-y+\frac{x^2+y^2}{4}+\frac{2xy}{3}\\
			&-\frac{(x+y-1)^4}{24x^6y^6}\,\Bigl(
			x^{10}-4x^9y+7x^8y^2-8x^7y^3+7x^6y^4-4x^5y^5\\
			&\hspace{4.3em}{}+7x^4y^6-8x^3y^7+7x^2y^8-4xy^9+y^{10}\Bigr).
		\end{aligned}
		\]
		The two formulas agree continuously along $x+y=1$, so $I$ is continuous on $T$. Integrating them over $T_0$ and $T_1$ gives
		\[
		\iint_T I(x,y)\,dx\,dy=-\frac{1}{144}-\frac{\pi^2}{36}+\frac{\log^2 2}{6}+\frac{19\log 2}{36}.
		\]
		Thus $\frac{12}{\pi^2}\iint_T I=B_{12}$. Since $N=\lfloor\sqrt n\rfloor$, Remark~\ref{rem:floor-sqrt-safe} gives $N^8=n^4+o(n^4)$ and $N^7\log N=o(n^4)$. Therefore
		\[
		S_{12}(n)=B_{12}n^4+o(n^4),
		\]
		as claimed.
	\end{proof}
	
	\subsection{Assembly}
	
	\begin{lemma}[covering identity]\label{lem:asym-cover}
		For every admissible quadruple $(u,v,a,b)$ with $u>v\ge 1$, $(u,v)=1$, $a,b\ge 1$, $au+bv\le n$, and $av+bu\le n$, one has either $u,v\le \sqrt n$ or $a,b\le \sqrt n$. Consequently,
		\[
		F(n)=S_1(n)+S_2(n)-S_{12}(n)+F_0(n).
		\]
		Here $F_0(n)$ is exactly the exceptional contribution of the three directions $(1,0)$, $(0,1)$, and $(1,1)$ that were removed from the primitive summations at the start of the appendix.
	\end{lemma}
	
	\begin{proof}
		If $u>\sqrt n$, then $au<n$ and $bu<n$, hence $a,b<n/u<\sqrt n$. Thus every admissible quadruple lies in $\{u,v\le \sqrt n\}\cup\{a,b\le \sqrt n\}$, and the formula for $F(n)$ follows by inclusion--exclusion.
	\end{proof}

	\begin{theorem}[assembly of the two-term asymptotic]\label{thm:asym-assembly-main}
		One has
		\[
		F(n)=A\,n^4\log n+B\,n^4+o(n^4),
		\]
		where
		\[
		A=\frac{4\log 2-1}{\pi^2}
		\]
		and
		\[
		B=-\frac{4\log 2-1}{6}\frac{\zeta'(2)}{\zeta(2)^2}+\frac{24(4\log 2-1)\gamma+72\log^2 2-76\log 2+1}{12\pi^2}-\frac14.
		\]
	\end{theorem}
	
	\begin{proof}
		By Lemma~\ref{lem:asym-cover} and the closed form from Section~\ref{sec:standard},
		\[
		F(n)=S_1(n)+S_2(n)-S_{12}(n)+F_0(n),\qquad F_0(n)=\frac13n^4+O(n^3).
		\]
		Applying Theorems~\ref{thm:asym-S1}, \ref{thm:asym-S2}, and~\ref{thm:asym-S12}, we obtain
		\[
		A=A_1+A_2=\frac{4\log 2-1}{\pi^2},\qquad B=B_1+B_2-B_{12}+\frac13.
		\]
		Note in particular that $S_{12}(n)$ contributes only at order $n^4$, so the entire $n^4\log n$ term comes from $S_1(n)+S_2(n)$. Equivalently, the logarithm arises from the one-parameter sums in $S_1$ and $S_2$, whereas the overlap $S_{12}$ already lives on the doubly truncated scale $u,v,a,b\ll \sqrt n$ and therefore has size only $n^4$. Substituting the explicit constants yields the stated formula for $B$.
	\end{proof}

\end{document}